\documentclass[twoside,11pt]{article}

\usepackage{blindtext}

\usepackage{jmlr2e}

\usepackage{amsfonts, amsmath, amssymb, parskip, bbm, booktabs}
\usepackage{xcolor, graphicx, tikz-cd, caption, subcaption}
\usepackage{algorithmic, algorithm}
\usepackage{verbatim}
\usepackage{mathrsfs}
\usepackage{authblk}
\newtheorem{assumption}{Assumption}
\newcommand{\law}{\textrm{Law}}
\newcommand{\bfY}{\mathbf{Y}}
\newcommand{\bfX}{\mathbf{X}}
\newcommand{\bfV}{\mathbf{V}}
\newcommand{\bfZ}{\mathbf{Z}}

\newcommand\daniel{}    
\long\def\daniel#1{#1}
\long\def\james#1{#1}

\usepackage{lastpage}

\ShortHeadings{The Two-System Paradigm}{Chok, Lee, Paulin, and Vasil}
\firstpageno{1}

\begin{document}

\title{Divide, Interact, Sample: The Two-System Paradigm}

\author[1]{\large{James Chok}}
\author[1]{\large{{Myung Won} Lee}}
\author[2]{\large{Daniel Paulin}}
\author[1]{\large{Geoffrey M. Vasil}}

\affil[1]{\normalsize School of Mathematics and Maxwell Institute for Mathematical Sciences, University of Edinburgh, United Kingdom}
\affil[2]{School of Physical and Mathematical Sciences, Nanyang Technological University, Singapore, \authorcr
  \{\tt james.chok, johnny.myungwon.lee, g.vasil\}@ed.ac.uk \authorcr
  \tt daniel.paulin@ntu.edu.sg}
\makeatletter
\renewcommand{\@starteditor}{}
\renewcommand{\@endeditor}{}
\makeatother
\editor{}

\maketitle

\begin{abstract}
Mean-field, ensemble-chain, and adaptive samplers have historically been viewed as distinct approaches to Monte Carlo sampling. In this paper, we present a unifying {two-system} framework that brings all three under one roof. In our approach, an ensemble of particles is split into two interacting subsystems that propose updates for each other in a symmetric, alternating fashion. For the memoryless two-system samplers, this cross-system interaction ensures that the finite ensemble has $\rho^{\otimes 2N}$ as its invariant distribution; for finite-adaptive variants, exact stationarity applies after the adaptation phase is frozen. The two-system construction reveals that ensemble-chain samplers can be interpreted as finite-$N$ approximations to an ideal mean-field sampler; conversely, it provides a principled recipe for discretizing mean-field Langevin dynamics into tractable parallel MCMC algorithms. The framework also connects naturally to adaptive single-chain methods: by replacing particle-based statistics with time-averaged statistics from a single chain, one recovers analogous adaptive dynamics in the long-time limit without requiring a large ensemble. We derive novel two-system versions of both overdamped and underdamped Langevin MCMC samplers within this paradigm. \james{Across synthetic benchmarks and real-world posterior inference tasks, these two-system samplers -- which use a single BCSS-2 integrator step per Metropolis--Hastings accept/reject, in contrast to the long-trajectory style of HMC/NUTS -- exhibit substantial performance gains over No-U-Turn Sampler baselines, achieving higher effective sample sizes per gradient evaluation and markedly higher wall-clock throughput. On higher-dimensional posteriors, the adaptive MAKLA-BCSS-2 methods remain stable and achieve substantially better per-gradient efficiency and wall-clock throughput than the NUTS variants in our benchmark suite.}
\end{abstract}


\section{Introduction}
Sampling from high-dimensional probability distributions is pivotal in modern Bayesian statistics \citep{Richardson2011, vandeSchoot2021, Zhang2024}. Analytical formulas for moments or normalizing constants are often infeasible to derive or compute directly in such settings. Monte Carlo sampling provides a practical alternative: by drawing representative samples, one can approximate the expectation of a posterior while incorporating prior information, and also quantify uncertainty via posterior intervals \citep{Jaynes2003, Gelman2013}.

Two broad classes of MCMC algorithms are widely used: (i) \textit{Discrete-time} methods based on the Metropolis-Hastings algorithm, and (ii) \textit{Continuous-time} methods defined by Langevin-type stochastic differential equations (SDEs). Both approaches yield Markov chains whose stationary distribution is the desired target distribution. The increasing availability of parallel computing has made it practical to run multiple chains simultaneously, which in turn has spurred the development of ensemble-based methods that leverage interactions among chains to improve convergence.

Modern sampling strategies often adapt their proposal mechanism using either information from other chains or from the chain's own history:
\daniel{\textit{Mean-field samplers} let each particle use information from the entire ensemble's empirical distribution; this scheme has $\rho$ as its invariant distribution only in the infinite-particle limit. \textit{Ensemble-chain samplers} approximate this idea at finite $N$ using statistics computed from a subset of the other chains; existing constructions either inherit the same $N\to\infty$ requirement or restore finite-$N$ invariance only at $\mathcal{O}(N)$ extra cost per step \citep{nusken2019, inigo_2020}.} \textit{Adaptive samplers} use a single chain and adjust its proposal parameters on the fly based on the chain's past samples (e.g., continually tuning the proposal covariance based on a running estimate). This adaptation breaks the Markov property, but convergence to the target can still be established under suitable conditions.

To unify these perspectives, we introduce the \textit{two-system} approach. This framework divides an ensemble of $N$ particles into two interacting subsystems, with each subsystem using the empirical distribution of the other as the basis for its proposals. For the memoryless two-system construction, this symmetric coupling guarantees finite-$N$ invariance, while finite-adaptive variants recover exact stationarity after the preconditioner is frozen. Starting from a mean-field Langevin SDE, we show that the two-system formulation naturally yields ensemble-chain samplers via finite-particle discretization of the continuous dynamics. By alternating updates between the two subsystems in discrete time (with appropriate Metropolis-Hastings corrections), we obtain parallel MCMC algorithms that preserve detailed balance with respect to $\rho$ and are well-suited for modern multi-core hardware. \daniel{The memoryless Coupled MALA / MAKLA-BCSS-2 samplers (Algorithms~\ref{algo:interacting_mala}--\ref{algo:interacting_makla}) preserve $\rho^{\otimes 2N}$ exactly at every step. For the finite-adaptive variants (Algorithms~\ref{algo:adaptive_mala}--\ref{algo:adaptive_makla}), exact $\rho^{\otimes 2N}$-stationarity applies to the post-adaptation fixed-kernel phase, once the preconditioner has been frozen at $T_{\mathrm{adapt}}$ (Remark~\ref{rmk:finite_adaptation}); during the adaptive warm-up phase the chain is history-dependent.}

This two-system perspective also clarifies the relationship between ensemble and mean-field methods. At each iteration, an ensemble-chain sampler uses finite-sample estimates (e.g., a sample covariance) in place of exact population quantities; as $N \to \infty$, these estimates converge to their mean-field counterparts and the particle distribution approaches $\rho$. In the long-time limit, the particles approximate independent draws from $\rho$, and the empirical statistics converge to the corresponding expectations under $\rho$. Similarly, an adaptive single-chain sampler can be viewed as the $N=1$ analog of a mean-field scheme: instead of averaging over many particles at a given time, it averages over many iterations of one chain. As the number of iterations grows, the running time-average of the statistic converges to the true expectation under $\rho$. In this way, adaptive samplers mirror the effect of ensemble methods, substituting temporal averaging for population averaging.

In summary, this work makes the following contributions:
\begin{itemize}
    \item We propose a unified \textbf{two-system sampling framework} that encompasses ensemble-chain MCMC, mean-field samplers, and adaptive samplers within a single coherent paradigm.
    \item Using this framework, we derive \textbf{novel two-system MCMC algorithms} for both overdamped and underdamped Langevin dynamics. Specifically, two-system versions of the Metropolis-Adjusted Langevin Algorithm (MALA) and a new BCSS-2 instantiation of the Metropolis-Adjusted Kinetic Langevin Algorithm (MAKLA) -- which we refer to throughout as \textbf{MAKLA-BCSS-2}, to distinguish it from the original Verlet-based MAKLA of \citet{BouRabee2024} -- in both ensemble-chain and adaptive forms. \emph{All variants use a single BCSS-2 integrator step per Metropolis--Hastings accept/reject ($L=1$), with persistent momentum (flipped on rejection); we never run multi-step Hamiltonian trajectories.}
    \item \james{Through extensive experiments, we provide \textbf{empirical evidence of substantial performance gains} over strong No-U-Turn Sampler (NUTS) \citep{hoffman2014no} baselines. In particular, adaptive and ensemble-chain MAKLA-BCSS-2 variants improve substantially over their static-preconditioned counterparts on challenging synthetic and real-data targets, with more than an order-of-magnitude reduction in worst-component gradient cost on the banana benchmark. On 45 different posterior models, the MAKLA-BCSS-2 family achieves the best aggregate per-gradient efficiency and wall-clock throughput among the twelve samplers compared, with consistent gains over Hessian-preconditioned NUTS and much larger improvements over default diagonal-mass NUTS variants.}
    \item \daniel{We propose a set of \textbf{practical refinements} -- per-step step-size randomization tuned by a high-acceptance ($\mathrm{acc}\geq 1-h/16$) criterion, a partial-reset adaptation schedule, Hessian-at-MAP rescaling, the higher-order BCSS-2 BABAB inner integrator of \citet{blanes2014} in place of the standard Verlet leapfrog, and a fused log-density / gradient kernel with cached values across consecutive steps -- that together let two-system MAKLA-BCSS-2 samplers compete with NUTS in both per-gradient efficiency and wall-clock throughput in our experiments.}
\end{itemize}

\daniel{\textbf{Paper overview.} Section~\ref{sec:background} reviews the relevant background on Langevin dynamics, mean-field samplers, ensemble-chain methods, and adaptive MCMC. Section~\ref{sec:two_system_mckean_vlasov} introduces the two-system construction for McKean--Vlasov equations, in which two interacting subsystems evolve using information from one another while preserving the desired invariant distribution. Section~\ref{sec:two_system_samplers} applies this framework to sampling, deriving finite-particle ensemble-chain algorithms and explaining how adaptive samplers arise by replacing population averages with time averages. Section~\ref{sec:ensemble_chain_mcmc} develops concrete Metropolis-adjusted two-system MALA and MAKLA-BCSS-2 algorithms, together with practical modifications for step-size randomization, adaptation, and ill-conditioned targets. Section~\ref{sec:simulations_experiments} evaluates these methods on synthetic benchmarks and on $45$ \texttt{posteriordb} posteriors against ensemble, adaptive, and NUTS baselines.}

\textit{Notation}: Throughout the paper, $\rho:\mathbb{R}^d \to [0,\infty)$ denotes the target probability density we wish to sample from. We assume $\rho(x)$ is bounded, absolutely continuous, and that $\log\rho(x)$ is $C^2$-smooth. We let $S^{d}_{++}$ denote the set of real $d\times d$ symmetric positive-definite matrices, $I_d$ the $d\times d$ identity matrix, $\mathcal{P}_2(\mathbb{R}^d)$ the space of probability distributions on $\mathbb{R}^d$ with finite second moments (with densities with respect to Lebesgue measure), and $\mathcal{B}(\mathbb{R}^d)$ the Borel $\sigma$-algebra on $\mathbb{R}^d$.

\section{\daniel{Related Work}}
Our work lies at the intersection of ensemble-chain MCMC, adaptive sampling, and mean-field dynamics. Prior methods have explored various combinations of these ideas, but we provide a systematic two-system framework that unifies and extends these lines of work with minimal computational overhead.

The well-known affine invariant ensemble sampler was proposed by \citet{goodman2010ensemble}. This is an efficient sampler for low-dimensional problems, but it does not scale as well in higher dimensions as Langevin-based methods.

\textbf{Overdamped Langevin Samplers.} \citet{inigo_2020} proposed a mean-field ensemble overdamped Langevin SDE that converges to the correct target only in the infinite-particle (mean-field) limit. \citet{nusken2019} later corrected the finite-$N$ invariant measure using a leave-one-out scheme, but at the cost of $\mathcal{O}(N)$ covariance computations per step. In contrast, our two-system formulation (Section \ref{sec:two_system_mckean_vlasov}) achieves the same correction using only \emph{two} covariance evaluations per step (independent of $N$). We further present finite-particle discretizations (Section \ref{sec:two_systems_continuous_time_samplers}) that preserve the target distribution without requiring the mean-field limit.

\textbf{Underdamped Langevin Samplers.} Recent works have introduced ensemble-based Metropolis-adjusted underdamped samplers that adapt the momentum refresh \citep{leimkuhler2018ensemble, hoffman22a_ghmc, durand23a}, following ensemble Hamiltonian Monte Carlo~\citep{Buchholz2021}. However, these adaptations are limited to changes in the discretization and do not modify the underlying SDE itself. By contrast, our approach directly adapts the continuous-time dynamics via a preconditioned underdamped Langevin process in the two-system framework. Additionally, whereas the above methods adjust the integration step size to cope with ill-conditioning, we mitigate poor conditioning through a randomized step size (Section~\ref{sec:step_size_randomization}) and a one-time preprocessing step (Section~\ref{sec:ill_conditioned_distribution}) that rescales the target based on local curvature.

\daniel{\textbf{Adaptive MCMC.} A separate line of work develops adaptive MCMC schemes that tune the proposal parameter from a single chain's history rather than from cross-chain interactions. Beginning with the adaptive Metropolis algorithm of \citet{Haario2001}, this literature \citep{Atchad2006, roberts2007, Liang2010, laitinen2024invitationadaptivemarkovchain} establishes the two foundational requirements -- geometric ergodicity of the frozen kernel for each fixed parameter value and diminishing adaptation along the parameter trajectory -- under which the resulting non-Markovian process still converges to the target. Practical Hamiltonian-style implementations (notably the dual-averaging step adaptation and three-phase window adaptation of \citet{hoffman2014no}) are a special case of this paradigm and underlie production tools such as \texttt{Stan} \citep{carpenter2017stan}. Our two-system construction is complementary to this line of work: where adaptive MCMC replaces population averaging with temporal averaging in a single chain, the two-system approach replaces population averaging with cross-subsystem averaging in a finite ensemble (Section~\ref{sec:two_system_samplers}); both reduce to the same mean-field limit but offer different trade-offs in finite-$N$ behavior, and our adaptive two-system schemes (Algorithms~\ref{algo:adaptive_mala}--\ref{algo:adaptive_makla}) combine the two by adapting at the subsystem level.} 

\section{Background}\label{sec:background}
As mentioned above, samplers can broadly be categorized into \textit{continuous-time methods} (based on Langevin-type SDEs) and \textit{discrete-time methods} (e.g., the Metropolis-Hastings algorithm). Continuous-time samplers leverage the fact that a Langevin SDE has $\rho(x)$ as its stationary distribution, meaning that as $t\to\infty$, the law of the SDE converges to $\rho(x)$. However, in practice, one must discretize the SDE to simulate it, and such discretizations typically introduce bias and no longer preserve $\rho(x)$ exactly. To correct this bias, the discretized SDE can be used as a proposal within a Metropolis-Hastings framework, yielding a Markov chain that converges to the correct target. More broadly, the Metropolis-Hastings algorithm provides a general recipe for constructing discrete-time samplers that converge to $\rho(x)$. In high-dimensional settings, however, naive proposals can lead to very low acceptance rates, resulting in poor mixing and inefficient exploration of the target.

\subsection{Discrete-Time Samplers}\label{sec:intro_discrete_time_samplers}
The classic Metropolis-Hastings algorithm \citep{Metropolis1953, Hastings1970} constructs a Markov chain $\{X_k\}_{k\ge 0}$ with stationary distribution $\rho(x)$. Each iteration consists of two steps: (i) Given the current state $X_k$, propose a new state $\widetilde X_{k+1} \sim Q_C(\cdot\mid X_k)$, where $Q_C(\cdot\mid X_k)$ is a proposal distribution parameterized by some matrix $C$ (e.g., a covariance); (ii) Accept the proposal with probability 
\begin{equation}
    A(X_k, \widetilde X_{k+1})\ =\ \min\!\left\{1,\ \frac{\rho(\widetilde X_{k+1})}{\rho(X_k)}\,\frac{Q_C(X_k\mid \widetilde X_{k+1})}{Q_C(\widetilde X_{k+1}\mid X_k)}\right\},\notag
\end{equation}
and set $X_{k+1}=\widetilde X_{k+1}$ if accepted (otherwise $X_{k+1}=X_k$). We denote by $P_C$ the Markov transition kernel obtained by composing the proposal step with the Metropolis--Hastings accept--reject correction, and write $X_{k+1}\sim P_C(X_k,\cdot)$. We refer to \citep{roberts_1996_aperiodic, roberts_1996_mala, Robert2004} for more details on conditions ensuring detailed balance (and hence convergence) for the proposal distribution $Q_C$ and acceptance ratio $A$.

\james{ For the remainder of this paper, we only consider the case in which the proposal distribution is parameterized by a covariance, or preconditioning, matrix. Concretely, we assume that the proposal parameter belongs to the admissible set
\begin{equation}\label{eq:admissible_set}
    \Gamma\ :=\ \Gamma_{\varepsilon,K_{\mathrm{cov}}}\ :=\ \big\{\Theta\in S^d_{++}\ :\ \varepsilon\, I_d\preceq \Theta\ \text{and}\ \lVert\Theta\rVert_{\mathrm{op}}\leq K_{\mathrm{cov}}\big\},
\end{equation}
for fixed constants $0<\varepsilon<K_{\mathrm{cov}}<\infty$, where $\lVert\cdot\rVert_{\mathrm{op}}$ denotes the operator (spectral) norm and $A\preceq B$ means $B-A$ is positive semi-definite. To map an empirical covariance $A\in S^d_{+}$ (positive semi-definite) into $\Gamma$, we use the \emph{cap-then-ridge} map
\begin{equation}\label{eq:projection}
    \Pi_{\Gamma}^{\varepsilon,K_{\mathrm{cov}}}(A)\ =\ \varepsilon\, I_d\ +\ \alpha(A)\,A,
    \qquad
    \alpha(A)\ =\
    \begin{cases}
        1, & \lVert A\rVert_{\mathrm{op}}=0,\\[4pt]
        \min\!\left(1,\ \dfrac{K_{\mathrm{cov}}-\varepsilon}{\lVert A\rVert_{\mathrm{op}}}\right), & \lVert A\rVert_{\mathrm{op}}>0.
    \end{cases}
\end{equation}
This first scales the empirical covariance component (only when its operator norm exceeds $K_{\mathrm{cov}}-\varepsilon$) and then adds the ridge $\varepsilon I_d$, in that order. By construction, for any $A\succeq 0$ both bounds hold simultaneously: $\Pi_{\Gamma}^{\varepsilon,K_{\mathrm{cov}}}(A)\succeq\varepsilon I_d$ since $\alpha(A)A\succeq 0$, and
\begin{equation*}
    \big\lVert\Pi_{\Gamma}^{\varepsilon,K_{\mathrm{cov}}}(A)\big\rVert_{\mathrm{op}}
    = \varepsilon + \alpha(A)\lVert A\rVert_{\mathrm{op}}
    \leq \varepsilon + (K_{\mathrm{cov}}-\varepsilon) = K_{\mathrm{cov}},
\end{equation*}
so $\Pi_{\Gamma}^{\varepsilon,K_{\mathrm{cov}}}(A)\in\Gamma$. The map only requires $\lVert A\rVert_{\mathrm{op}}$, which can be computed exactly or approximated in practice by power iteration (each iteration is a single dense matrix-vector product); using a certified upper bound in place of $\lVert A\rVert_{\mathrm{op}}$ preserves the theoretical guarantee. The set $\Gamma$ is compact as a subset of the finite-dimensional space of symmetric $d\times d$ matrices and is contained in $S^d_{++}$; this compactness is the standard regularity used in adaptive-MCMC convergence theory \citep{Atchad2006, roberts2007}, and it is the formulation used in our implementation of Algorithms~\ref{algo:adaptive_mala}--\ref{algo:adaptive_makla}.}

\james{\textit{Mean-field samplers} \citep{clarte2022} modify the proposal distribution so that the parameter $C$ is now a (bounded and Lipschitz continuous) functional of the current law of the chain $C:\mathcal{P}_2(\mathbb{R}^d)\to\Gamma$.} The chain at iteration $k$ then evolves according to the nonlinear Markov transition kernel
\begin{equation}\label{eq:discrete_mean_field_dynamics}
    X_{k+1}\sim P_{C(\mu_k)}(X_k|\cdot),\qquad  \text{where}\qquad  \mu_k = \law(X_k),
\end{equation}
with proposals $\widetilde X_{k+1}\sim Q_{C(\mu_k)}(\cdot \mid X_k)$. Since $\mu_k$ is generally unknown in closed form, one typically approximates it by simultaneously evolving $N$ particles $\bfX_k=\{X_k^i\}_{i=1}^N$. Concretely, the finite-particle approximation then evolves the particles via
\begin{equation}
    X_{k+1}^i\ \sim\ P_{C(\delta_{\bfX_k})}(X_k^i,\, \cdot),\qquad \delta_{\bfX_k}\ =\ \frac{1}{N}\sum_{i=1}^N \delta_{X_k^i},\notag
\end{equation}
where $\delta_x$ denotes a Dirac at $x$, with proposals $\widetilde X_{k+1}^i\sim Q_{C(\delta_{\bfX_k})}(\cdot|X_k^i)$ for $i=1,\ldots, N$. Then, as $N\to\infty$, this particle discretization recovers the mean-field dynamics \eqref{eq:discrete_mean_field_dynamics}, for which $\rho$ is the invariant measure. However, for any fixed finite $N$, the particle system generally does \emph{not} have $\rho$ as its stationary distribution \citep{clarte2022, sprungk2023}. 

\textit{Ensemble-chain methods} evolve $N>1$ interacting Markov chains simultaneously. Again, let  $\bfX_k=\{X_k^i\}_{i=1}^N$ denote the collection of $N$ chains at time $k$. We define a data-dependent statistic based on a designated subset $S \subset \{1,\ldots, N\}$ of ``reference'' particles:
\begin{equation}
    \widehat C(S)\ =\ C(\delta_{\{X_k^i\}_{i\in S}}),\notag
\end{equation}
where $S$ indexes the subset of particles that will be held \emph{fixed} during the next update.  Using this statistic, we propose updates for the complementary set $S^c=\{1,\ldots, N\}\setminus S$ as 
\begin{equation}
    X_{k+1}^i \sim P_{\widehat C(S)}(X_k^i,\, \widetilde X_{k+1}^i), \qquad \text{with}\qquad \widetilde X_{k+1}^i \sim Q_{\widehat C(S)}(\cdot\mid X_k^i), \quad \forall\, i\in S^c,\notag
\end{equation}
and we set $X_{k+1}^i = X_k^i$ for all $i \in S$. Importantly, keeping the particles in $S$ fixed (i.e., excluding self-interaction) is essential to ensure that the joint state $\bfX_k$ has $\rho^{\otimes N}$ as its invariant distribution \citep{nusken2019}. This exclusion principle mirrors the leave-one-out idea in continuous-time ensemble samplers (see Section \ref{sec:intro_continuous_time_samplers}) and guarantees detailed balance for the combined ensemble update.

\james{\textit{Adaptive samplers} use a single chain but update proposal parameters ``on the fly'' based on the chain's history. Concretely, at iteration $k$, we compute
\begin{align*}
    \Theta_{k}\ &=\ \left(1-\frac{1}{k}\right)\Theta_{k-1} +\frac{1}{k}\Pi_{\Gamma}^{\varepsilon,K_{\mathrm{cov}}}\left((X_k-\mu_{k-1})(X_k-\mu_{k-1})^\top\right),\\
    \mu_k\ &=\ \left(1-\frac{1}{k}\right)\mu_{k-1} +\frac{1}{k}X_k,
\end{align*}
so that $\Theta_k$ is a running estimate of the covariance of the historical samples. Although $\Theta$ and $C:\mathcal{P}_2(\mathbb{R}^d)\to\Gamma$ represent the same proposal parameter, we write $\Theta$ to emphasize that the parameter is estimated from the historical samples.} 

Then the next proposal is generated via
\begin{equation*}
    X_{k+1}\sim P_{\Theta_k}(X_k,\,\widetilde X_{k+1}),\qquad\text{with}\qquad \widetilde X_{k+1}\sim Q_{\Theta_k}(\cdot |X_k).
\end{equation*}
Because the proposal now depends on past samples, the resulting process is no longer Markovian. To ensure convergence to $\rho$, two key conditions are typically required:
\begin{itemize}
    \item \textit{Geometric ergodicity.} For every fixed parameter value $\gamma\in \Gamma$,
    \begin{equation*}
    \lVert P_\gamma^n(x,\cdot) - \rho(\cdot)\rVert_{\text{TV}}\ \leq\ c\, r^n,
    \end{equation*} 
    for some constants $c>0$, $0<r<1$. Here $P^n_\gamma$ is the $n$-step transition kernel and $\|\cdot\|_{\text{TV}}$ denotes the total variation norm. Intuitively, this says that if the algorithm is run without adaptation (parameter held fixed at $\gamma$), it should converge geometrically fast to the target.
    \item \textit{Diminishing adaptation.} 
    \begin{equation*}
    \lim_{n\to\infty}\ \sup_{x \in \mathbb{R}^d}\, \lVert P_{\gamma_{n+1}}(x,\cdot) - P_{\gamma_n}(x,\cdot)\rVert_{\text{TV}}\ =\ 0,
    \end{equation*} 
    in probability. This guarantees that the magnitude of adaptation vanishes as the chain progresses.
\end{itemize}
Under these conditions (plus some additional technical assumptions), the adaptive algorithm remains ergodic and satisfies a law of large numbers: for any bounded measurable test function $h:\mathbb{R}^d\to\mathbb{R}$,
\begin{equation}
    \frac{1}{K}\sum_{k=1}^K h(X_k)\ \to\ \int h(x)\,\rho(x)\,dx,\notag
\end{equation}
in probability as $K\to\infty$ \citep{Liang2010}. (A strong law of large numbers can also be established under stronger assumptions on $\rho$ and the kernels \citep{laitinen2024invitationadaptivemarkovchain}.) For more details on adaptive MCMC, we refer the reader to \citep{Liang2010, laitinen2024invitationadaptivemarkovchain}.

\subsection{Continuous-Time Samplers}\label{sec:intro_continuous_time_samplers}
Continuous-time samplers evolve a stochastic process whose stationary law is the target distribution $\rho$. The most common examples are based on the overdamped and underdamped Langevin dynamics. The overdamped Langevin SDE \citep{Langevin1908, Lemons1997} is given by
\begin{equation}\label{eq:overdamped_langevin}
    dX_t = C_0\,\nabla \log\rho(X_t)\, dt \;+\; \sqrt{2\,C_0}\, dW_t\,,
\end{equation}
where $C_0 \in S^d_{++}$ (often $C_0 = I_d$), $X_t\in\mathbb{R}^d$, and $W_t$ is standard $d$-dimensional Brownian motion. The underdamped Langevin (kinetic Langevin) dynamics introduces an auxiliary velocity variable $V_t$ and evolves as
\begin{align}\label{eq:underdamped_langevin}
    \begin{split}
    dV_t &= -\alpha\, V_t\, dt \;+\; \gamma\,C_0^{1/2}\, \nabla \log\rho(X_t)\, dt \;+\; \sqrt{2\alpha \gamma}\, dW_t, \\
    dX_t &= C_0^{1/2}\,V_t\, dt\,,
    \end{split}
\end{align}
with constants $\alpha,\gamma>0$ and initial conditions $X_0, V_0 \in \mathbb{R}^d$.

Under our assumptions on $\rho(x)$, if $\rho$ further satisfies a strong log-concavity condition (i.e., $\log\rho(x)$ has a globally bounded Hessian), then the law of $X_t$ under \eqref{eq:overdamped_langevin} converges to $\rho$, and the law of $(X_t,V_t)$ under \eqref{eq:underdamped_langevin} converges to the product measure $\rho(x)\,g(v)$, where $g(v)$ is the standard Gaussian density on velocities \citep{Pavliotis2014}.

\textit{Mean-Field Sampler (Continuous-Time).} \citet{inigo_2020} proposed a mean-field variant of the overdamped Langevin SDE. This sampler evolves the nonlinear SDE
\begin{equation}\label{eq:mean_field_overdamped}
    dX_t = C(\mu_t)\,\nabla \log\rho(X_t)\, dt \;+\; \sqrt{2\,C(\mu_t)}\, dW_t,
\end{equation}
where $\mu_t = \law(X_t)$ and $C(\mu_t)$ denotes (for instance) the covariance matrix of $\mu_t$. An analogous mean-field extension exists for the underdamped Langevin case. Under appropriate conditions, the law of $X_t$ in \eqref{eq:mean_field_overdamped} still converges to $\rho$ as $t\to\infty$.

As in the discrete-time setting, to simulate \eqref{eq:mean_field_overdamped} one must use a finite-particle approximation, evolving particles according to 
\begin{equation}
    dX_t^i = C(\delta_{\mathbf{X}_t})\,\nabla \log\rho(X_t^i)\, dt \;+\; \sqrt{2\,C(\delta_{\mathbf{X}_t})}\, dW_t^i, \qquad i \in \{1,\ldots, N\},\notag
\end{equation}
with each $W_t^i$ an independent Brownian motion. 

\textit{Ensemble-Chain Sampler (Continuous-Time).} \citet{nusken2019} observed that the naive particle approximation above does not preserve $\rho^{\otimes N}$ as an invariant joint law. They proposed the modified dynamics
\begin{equation}
    dX_t^i\ =\ \widehat C(\{i\}^c)\,\nabla\log\rho(X_t^i)\, dt \;+\; \sqrt{2\,\widehat C(\{i\}^c)}\, dW_t^i\,,\quad\text{for}\quad i=1,\ldots, N,\notag
\end{equation}
where $\{i\}^c=\{1,\ldots, N\}\setminus \{i\}$, and $\widehat C(\{i\}^c)$ denotes the covariance computed over all particles except the $i$th. This exclusion of self-interaction guarantees that $\rho^{\otimes N}$ remains invariant, directly paralleling the discrete-time case.

\textit{Adaptive Samplers (Continuous-Time).} Although less explored, adaptive Langevin dynamics have been proposed in recent work \citep{Kim2021, leimkuhler2025langevinsamplingalgorithminspired}. These methods dynamically modify aspects of the SDE (typically the step size or preconditioner) based on the trajectory's history. In this paper, we focus on discrete-time adaptive samplers, which are more straightforward to analyze theoretically and implement in practice.

\section{A Two-System Approach for McKean--Vlasov Equations}\label{sec:two_system_mckean_vlasov}
We now develop the two-system approach for general continuous-time mean-field SDEs (McKean--Vlasov equations), of which the overdamped and underdamped Langevin equations are special cases. The key idea is to consider two coupled systems of particles such that each system ``drives'' the evolution of the other. In a Langevin setting, this symmetric coupling is constructed to preserve the target distribution and to directly connect the behavior of the finite-$N$ particle system with its $N\to\infty$ mean-field limit.

The mean-field Langevin equation \eqref{eq:mean_field_overdamped} is an example of a more general McKean--Vlasov SDE, which takes the form:
\begin{equation}\label{eq:mckean_vlasov_sde}
    d\overline{X}_t\ =\ b(t, \overline{X}_t, \mu_t)\, dt\ +\ \sigma(t, \overline{X}_t, \mu_t)\, dW_t, \qquad t\in[0,T],
\end{equation}
where $\mu_t = \law(\overline{X}_t)$, and $b:[0,T]\times \mathbb{R}^d \times \mathcal{P}(\mathbb{R}^d)\to \mathbb{R}^d$ and $\sigma:[0,T]\times \mathbb{R}^d \times \mathcal{P}(\mathbb{R}^d)\to \mathbb{R}^{d\times d}$ are given drift and diffusion functions. To ensure existence and uniqueness of strong solutions \citep{sznitman_1991, Carmona2016}, we assume $b$ and $\sigma$ satisfy:
\begin{assumption}[Linear growth]
\label{eq:mckean_vlasov_assumption_1}
For every finite $T>0$, there exists a constant $L>0$ such that
\begin{equation}
    |b(t,x,\mu)|^2+|\sigma(t,x,\mu)|^2\leq L\left(1+|x|^2+\int_{\mathbb{R}^d}|u|^2\mu(du)\right),
\end{equation}
with norms understood in Euclidean and Frobenius.
\end{assumption}
\begin{assumption}[Lipschitz continuity]\label{eq:mckean_vlasov_assumption_2}
    There exists a constant $L>0$ such that for all $t\in[0,T]$, $x_1,x_2\in\mathbb{R}^d$, and $\mu_1,\mu_2\in\mathcal{P}_2(\mathbb{R}^d)$,
    \begin{equation}
    \begin{aligned}
    &\lvert b(t, x_1,\mu_1) - b(t, x_2,\mu_2)\rvert^2\ +\ \lvert\sigma(t, x_1,\mu_1) - \sigma(t, x_2,\mu_2)\rvert^2 \\
    &\qquad\qquad\leq\ L\,\big(\lvert x_1 - x_2\rvert^2 + \mathcal{W}_2^2(\mu_1,\mu_2)\big)\,,
    \end{aligned}
    \end{equation}
    where $\mathcal{W}_2$ denotes the 2-Wasserstein distance.
\end{assumption}

The classical particle approximation to \eqref{eq:mckean_vlasov_sde} is given by the system of $N$ interacting SDEs:
\begin{equation}
    dX_t^j\ =\ b(t, X_t^j, \delta_{\bfX_t})\, dt\ +\ \sigma(t, X_t^j, \delta_{\bfX_t})\, dW_t^j, \qquad j=1,\dots,N,\notag
\end{equation}
with $W_t^1,\ldots,W_t^N$ independent Brownian motions and $\bfX_t = \{X_t^j\}_{j=1}^N$. As $N\to\infty$, this interacting particle system converges to the mean-field equation \eqref{eq:mckean_vlasov_sde} in the sense of \emph{propagation of chaos} \citep{sznitman_1991, chaintron2022}: the particles become asymptotically `independent' and each follows the law of the McKean--Vlasov solution $\overline{X}_t$.

\textbf{The Two-System Approach.}\\
Following the two-system paradigm, we introduce a coupled SDE on an augmented state $Z_t = (X_{t}, Y_{t}) \in \mathbb{R}^{2d}$, where $X_{t}, Y_{t} \in \mathbb{R}^d$. We define
\begin{equation}\label{eq:coupled_mckean_vlasov}
    dZ_t\ =\ \mathbf{b}(t, Z_t, \pi_t)\, dt\ +\ \boldsymbol{\sigma}(t, Z_t, \pi_t)\, dW_t,
\end{equation}
where $\pi_t = \law(Z_t)$ and $W_t$ is a Brownian motion in $\mathbb{R}^{2d}$. The drift and diffusion are constructed as
\begin{equation}
    \mathbf{b}(t, Z, \pi)\ =\ \begin{pmatrix}
        b(t, X,\, p_2\!\circ \pi) \\ 
        b(t, Y,\, p_1\!\circ \pi)
    \end{pmatrix}, 
    \qquad 
    \boldsymbol{\sigma}(t, Z, \pi)\ =\ \begin{pmatrix}
        \sigma(t, X,\, p_2\!\circ \pi) & 0\\
        0 & \sigma(t, Y,\, p_1\!\circ \pi)
    \end{pmatrix},
\end{equation}
with $p_1$ and $p_2$ denoting the projection maps from a distribution on $\mathbb{R}^{2d}$ to its first and second $\mathbb{R}^d$-marginals, respectively. In words, the two-system SDE \eqref{eq:coupled_mckean_vlasov} consists of two processes $X_{t}$ and $Y_{t}$, each evolving according to the original drift/diffusion coefficients $b,\sigma$ evaluated at the empirical law of the \emph{other} system. 

Under Assumptions~\ref{eq:mckean_vlasov_assumption_1} and \ref{eq:mckean_vlasov_assumption_2}, one can show that the coupled SDE \eqref{eq:coupled_mckean_vlasov} inherit the same growth at origin bound and Lipschitz continuity, and thus have a unique strong solution on $[0,T]$ (proof in Appendix~\ref{sec:existence_and_uniquness}). Moreover, if $X_0$ and $Y_0$ are independent and identically distributed with common law $\mu_0$, then $\law(X_t)=\law(Y_t)=\mu_t$ for $0\leq t\leq T$, where $\mu_t$ is the unique solution law of the original McKean--Vlasov SDE \eqref{eq:mckean_vlasov_sde} with initial law $\mu_0$ (proof in Appendix~\ref{sec:convergence_same_distribution}). 

Next, we consider the natural particle approximation of the two-system SDE. We introduce a coupled $2N$-particle system:
\begin{equation}\label{eq:mckean_vlasov_finite_particle_approximation}
    \begin{split}
        dX_t^j\ &=\ b(t, X_t^j, \delta_{\bfY_t})\, dt\ +\ \sigma(t, X_t^j, \delta_{\bfY_t})\, dW_t^{X,j},\\
        dY_t^j\ &=\ b(t, Y_t^j, \delta_{\bfX_t})\, dt\ +\ \sigma(t, Y_t^j, \delta_{\bfX_t})\, dW_t^{Y,j}, \qquad j\in \{1,\ldots, N\},
    \end{split}
\end{equation}
where $\bfX_t = \{X_t^j\}_{j=1}^N$ and $\bfY_t = \{Y_t^j\}_{j=1}^N$. Each particle in one subsystem therefore interacts only through the empirical distribution of the opposite subsystem. This construction preserves the usual mean-field structure while decoupling the empirical measure used to update each subsystem. As shown in Appendix~\ref{sec:propagation_of_chaos}, the approximation~\eqref{eq:mckean_vlasov_finite_particle_approximation} satisfies the standard propagation-of-chaos property as $N\to\infty$.

\begin{remark}
    Our two-system construction is related to the conditional propagation of chaos used in maximum-likelihood estimation \citep{Carmona2016, johnston2024, kuntz_2023}, where one ensemble $\bfX_t$ approximates an integral that drives the dynamics of a separate process $Y^0_t$. The structural difference is that conditional propagation of chaos couples each $X^i_t$ to its own ensemble (so the $X^i$'s remain mutually correlated), whereas our two-system construction uses the empirical law of $\bfY_t$ to drive $\bfX_t$ and vice versa, yielding conditional independence within each subsystem.
\end{remark}

\section{A Two-System Approach for Samplers}\label{sec:two_system_samplers}
\james{
For the rest of this paper, unless otherwise stated, we consider a restricted class of McKean--Vlasov equations
\begin{equation}\label{eq:restricted_mckean_vlasov_sde}
    d\overline{X}_t\ =\ C(\mu_t)b(\overline{X}_t)\, dt\ +\ \sqrt{2C(\mu_t)}\, dW_t, \qquad t\in[0,T],
\end{equation}
where $\mu_t=\law(\overline{X}_t)$, and $C(\mu)$ is Lipschitz continuous, bounded, and positive definite for all $\mu\in \mathcal{P}_2(\mathbb{R}^d)$. In particular, $b(\overline{X}_t)$ no longer depends on $t$ and the law of $\overline{X}_t$. As such, for $b(x)=-\nabla V(x)$, the invariant distribution takes the form $\rho(x)\propto\exp(-V(x))$. Thus, the two-system framework introduces a pair of interacting Markov chains, where each chain generates proposals based on the empirical law of the other. When applied to a continuous-time mean-field sampler, this coupling preserves $\rho$ as the invariant distribution in both the finite-particle setting and the infinite-particle limit. By adding Metropolis--Hastings acceptance steps to an alternating (staggered) time-discretization of this coupled system, we obtain an ensemble-chain sampler that exactly maintains the target distribution $\rho$.

The continuous-time statements that follow assume an abstract bounded, Lipschitz-on-$(\mathcal{P}_2,\mathcal{W}_2)$, uniformly positive-definite preconditioner $C$. The empirical-covariance preconditioners used in the discrete-time Metropolis-adjusted algorithms of Section~\ref{sec:ensemble_chain_mcmc} need not satisfy this global $\mathcal{W}_2$-Lipschitz condition on all of $\mathcal{P}_2(\mathbb{R}^d)$; their exact $\rho^{\otimes 2N}$-invariance follows instead from conditional detailed balance in the discrete-time kernels (Appendix~\ref{appendix:ensemble_chain_proof}), independently of the continuous-time well-posedness theory.

This connection demonstrates that ensemble-chain samplers approximate expectations under $\rho$ by leveraging particle interactions. In the limit as $N \to \infty$ (and as time $t\to\infty$), the ensemble statistics converge to their mean-field counterparts. This also motivates the use of adaptive samplers, which replace the whole ensemble with a single chain and rely on long-run time averages to recover the same mean-field behavior.}

\subsection{\daniel{Two-System Continuous-Time Samplers}}\label{sec:two_systems_continuous_time_samplers}
To illustrate, consider the mean-field overdamped Langevin example. The two-system construction leads to a coupled SDE on $Z_t = \{X_t, Y_t\}$ (with $X_t, Y_t \in \mathbb{R}^d$):
\begin{equation}\label{eq:overdamped_langevin_two_system}
    \begin{split}
    dX_t\ &=\ C(p_2\!\circ\mu_t)\,\nabla\log\rho(X_t)\, dt \;+\; \sqrt{2\,C(p_2\!\circ \mu_t)}\, dW_t^X,\\
    dY_t\ &=\ C(p_1\!\circ\mu_t)\,\nabla\log\rho(Y_t)\, dt \;+\; \sqrt{2\,C(p_1\!\circ \mu_t)}\, dW_t^Y,
    \end{split}
\end{equation}
for $t\in[0,T]$, where $\mu_t = \law(Z_t)$ and we assume $C(\mu)$ remains bounded and positive definite for all $\mu\in\mathcal{P}(\mathbb{R}^d)$. Under mild regularity conditions on $\nabla \log\rho$, the two-system McKean--Vlasov equation \eqref{eq:overdamped_langevin_two_system} admits a unique strong solution (e.g., if $\log\rho(x)$ is strictly concave outside a large ball). 

Applying the finite-particle approximation from the previous section, we obtain the following $2N$-particle system for the two coupled overdamped processes:
\begin{equation}\label{eq:continuous_cisl}
    \begin{split}
    dX^j_t\ &=\ C(\delta_{\bfY_t})\, \nabla\log\rho(X^j_t)\, dt \;+\; \sqrt{2\,C(\delta_{\bfY_t})}\, dW^{X,j}_t,\\
    dY^j_t\ &=\ C(\delta_{\bfX_t})\, \nabla\log\rho(Y^j_t)\, dt \;+\; \sqrt{2\,C(\delta_{\bfX_t})}\, dW^{Y,j}_t,
    \end{split}
\end{equation}
for $j=1,\dots,N$. By the result of Section \ref{sec:two_system_mckean_vlasov}, as $N\to\infty$, this system converges to the mean-field dynamics \eqref{eq:overdamped_langevin_two_system}. Crucially, the two-system coupling also ensures exact invariance of the target at finite $N$. In particular, as we prove below (Lemma \ref{lemma:continuous_cisl}), $\rho^{\otimes 2N}$ is an invariant density of the coupled $2N$-particle system \eqref{eq:continuous_cisl} for any $N$. In other words, even before taking $N\to\infty$, the two-system construction preserves the desired target distribution by design.

\begin{lemma}\label{lemma:continuous_cisl}
    Assume that $C:\mathcal{P}_2(\mathbb{R}^d)\to S^d_{++}$ is bounded, Lipschitz on $(\mathcal{P}_2(\mathbb{R}^d),\mathcal{W}_2)$, and uniformly positive-definite (so that $S(\bfX,\bfY)$ defined below is uniformly positive-definite on $\mathbb{R}^{2Nd}$), and that $\nabla\log\rho$ is locally Lipschitz with at most linear growth. Then the $2N$-particle system \eqref{eq:continuous_cisl} admits a unique strong solution, and the product density $\rho^{\otimes 2N}$ is a stationary density of the system.
\end{lemma}
\begin{proof}
    Concatenate the two ensembles into $\bfZ=(\bfX,\bfY)\in\mathbb{R}^{2Nd}$ and write the empirical measures $\mu^N_{\bfX}=N^{-1}\sum_j\delta_{X^j}$, $\mu^N_{\bfY}=N^{-1}\sum_j\delta_{Y^j}$. Then \eqref{eq:continuous_cisl} reads 
    $$d\bfZ_t=S(\bfZ_t)\nabla\log\rho^{\otimes 2N}(\bfZ_t)\,dt+\sqrt{2S(\bfZ_t)}\,d\mathbf{W}_t,$$
    with the block-diagonal lift
    \begin{equation*}
        S(\bfX,\bfY) = \begin{pmatrix}
            I_N\otimes C(\mu^N_{\bfY}) & 0\\
            0 & I_N\otimes C(\mu^N_{\bfX})
        \end{pmatrix}.
    \end{equation*}
    Existence and uniqueness follow from standard interacting-particle SDE theory \citep{sznitman_1991, Carmona2016}. Under the boundedness, local regularity, and uniform positive-definiteness hypotheses, the It\^o generator is
    \begin{equation*}
        \mathcal{L} f \;=\; S\,\nabla\log\rho^{\otimes 2N}\!\cdot\!\nabla f \;+\; \mathrm{tr}\!\big(S\,\nabla^2 f\big)
        \;=\; \big(\rho^{\otimes 2N}\big)^{-1}\,\nabla\!\cdot\!\big(\rho^{\otimes 2N}\,S\,\nabla f\big) \;-\; (\nabla\!\cdot\! S)\!\cdot\!\nabla f,
    \end{equation*}
    so if $\nabla\!\cdot\! S=0$ the generator is in symmetric divergence form with respect to $\rho^{\otimes 2N}$, which implies stationarity (see Proposition~4.2 of \citet{nusken2019} and \citealp[Sec.~4]{Pavliotis2014}). The upper-left block of $S$ depends only on $\bfY$ but has $\bfX$-row-indices (and vice versa for the lower-right block), so $\partial_{Z_j}S_{ij}\equiv 0$ for all $i,j$, giving $\nabla\cdot S\equiv 0$ and stationarity of $\rho^{\otimes 2N}$.
\end{proof}

This result shows that the two-system approach provides a principled way to transform a continuous-time mean-field sampler into a finite-particle sampler while maintaining $\rho$ as the invariant measure for any ensemble size.

To extend the construction to the underdamped Langevin \eqref{eq:underdamped_langevin}, introduce two coupled copies $(X_t,V_t)$ and $(Y_t,U_t)$ and denote 
\begin{equation*}
\Pi_t=\law(X_t,V_t,Y_t,U_t),\qquad \mu_t^X=p_X\circ\Pi_t,\quad \mu_t^Y=p_Y\circ\Pi_t
\end{equation*}
as the position marginals. The \emph{two-system kinetic Langevin} is
\begin{equation}\label{eq:two_sys_udl}
\begin{split}
dV_t &= -\alpha V_t\,dt \;+\; \gamma\,C(\mu_t^Y)^{1/2}\nabla\log\rho(X_t)\,dt \;+\; \sqrt{2\alpha\gamma}\,dW_t^V,\\
dX_t &= C(\mu_t^Y)^{1/2} V_t\,dt,\\
dU_t &= -\alpha U_t\,dt \;+\; \gamma\,C(\mu_t^X)^{1/2}\nabla\log\rho(Y_t)\,dt \;+\; \sqrt{2\alpha\gamma}\,dW_t^U,\\
dY_t &= C(\mu_t^X)^{1/2} U_t\,dt,
\end{split}
\end{equation}
again with cross-preconditioning via the other subsystem's \emph{position} law. When $C(\cdot)\equiv C_0$ is constant, \eqref{eq:two_sys_udl} reduces to two independent preconditioned kinetic Langevin processes whose marginal in $X$ has $\rho$ as the position invariant distribution by standard underdamped Langevin theory \citep{Pavliotis2014}. We treat \eqref{eq:two_sys_udl} primarily as a heuristic motivation for the discrete-time MAKLA-BCSS-2-based two-system algorithms of Section~\ref{sec:ensemble_chain_mcmc}; a rigorous well-posedness, finite-particle propagation-of-chaos, and invariance analysis for the underdamped two-system construction (analogous to the overdamped results of Lemmas~\ref{lemma:continuous_cisl} and~\ref{lem:inherited}) is more delicate because the diffusion coefficient is degenerate in $X$, and is left to future work. In the discrete-time setting, exact $\rho$-invariance for the two-system MAKLA-BCSS-2 samplers presented below follows directly from the per-particle Metropolis--Hastings correction together with the cross-coupling structure (see Appendix~\ref{appendix:ensemble_chain_proof}), and does not rely on a continuous-time analogue.

\subsection{\daniel{Two-System Discrete-Time Samplers}}\label{sec:two_systems_discrete_time_samplers}
To obtain a practical MCMC algorithm, we discretize the two-system dynamics in time using an alternating (staggered) time-stepping scheme. Let $h>0$ denote the step size, and suppose that the coupled state $(X_t,Y_t)$ is available at time $t$. One step of the discretized two-system dynamics is given by
\begin{equation}
    \begin{split}
    X_{t+h}\ &=\ X_t \;+\; h\,C(\mu_t^Y)\,\nabla\log\rho(X_t)\;+\; \sqrt{2h\,C(\mu_t^Y)}\, \xi^X_t,\\
    Y_{t+h}\ &=\ Y_t \;+\; h\,C(\mu_{t+h}^X)\,\nabla\log\rho(Y_t)\;+\; \sqrt{2h\,C(\mu_{t+h}^X)}\, \xi^Y_t,
    \end{split}
\end{equation}
where $\mu_t^Y$ is the empirical law of the $Y$-subsystem at time $t$, $\mu_{t+h}^X$ is the empirical law of the updated $X$-subsystem, and $\xi^X_t,\xi^Y_t$ are independent standard Gaussian random variables.

Thus, the $X$-particles are first updated using the current statistics of the $Y$-subsystem. The $Y$-particles are then updated using the newly computed statistics of the $X$-subsystem. This produces the staggered sequence of proposals $(X_t,Y_t)\to (X_{t+h}, Y_t)\to (X_{t+h}, Y_{t+h})$. Each proposal is subsequently corrected by a Metropolis--Hastings accept-reject step, ensuring that the resulting Markov chain preserves the desired target distribution.

In the finite-particle setting, this becomes an ensemble-chain sampler. The state space $\bfZ_t=(\bfX_t,\bfY_t)$ is updated by alternating between the two subsystems:
\begin{align}
    &\bfX_{t+1}\ \sim\ \prod_{i=1}^NP_{C(\delta_{\bfY_t})}(X_t^i,\, \cdot),\qquad \text{with}\qquad \bfY_{t+1}\ \sim\ \prod_{i=1}^NP_{C(\delta_{\bfX_{t+1}})}(Y_t^i,\, \cdot),
\end{align}
yielding $\bfZ_{t+1}=\{\bfX_{t+1}, \bfY_{t+1}\}$. Thus, one ensemble is advanced using the empirical distribution of the other ensemble, after which the roles are reversed. This is precisely the structure of an ensemble-chain MCMC method: each particle group supplies the empirical statistics used to construct proposals for the other group.

Provided that each kernel $P_{C(\mu)}$ satisfies detailed balance with respect to the target density $\rho$, the alternating scheme preserves the product measure $\rho^{\otimes 2N}$; see Appendix~\ref{appendix:ensemble_chain_proof}. Passing to the mean-field limit $N\to\infty$ then recovers the two-system version of the discrete-time mean-field dynamics in~\eqref{eq:discrete_mean_field_dynamics}.

\james{\begin{remark}
    \citet{sprungk2023} remove the finite-particle bias of the discrete-time mean-field samplers of \citet{clarte2022} by applying a Metropolis correction to the full ensemble. Given $\bfX_k$, an ensemble with $N$ particles, they propose $\widetilde{\bfY}_k \sim Q(\bfX_k,\cdot)$ and accept the entire ensemble with probability
    \begin{equation*}
        A(\bfX_k,\widetilde{\bfY}_k)
        =
        \min\left\{
        1,
        \frac{\rho^{\otimes N}(\widetilde{\bfY}_k)}
             {\rho^{\otimes N}(\bfX_k)}
        \frac{Q(\bfX_k \mid \widetilde{\bfY}_k)}
             {Q(\widetilde{\bfY}_k \mid \bfX_k)}
        \right\}.
    \end{equation*}
    This is an $Nd$-dimensional block update: the target distribution $\rho^{\otimes N}({\bfX})$ is a product over all particles, so acceptance typically deteriorates with $N$ unless the step size is reduced. Their Gibbs-type variant updates only a subset of particles, permitting larger steps but reducing movement per iteration; the acceptance ratio still contains a product over the updated subset.

    Our two-system construction removes the finite-particle bias differently. Each subsystem defines proposals for the other, and conditional on the opposite subsystem, particles in the updated subsystem are accepted or rejected individually. Thus, each Metropolis ratio involves a single $d$-dimensional particle update rather than an ensemble-level product, avoiding the acceptance degradation associated with large ensemble sizes.
\end{remark}}

\subsection{Adaptive Samplers}
In the mean-field Langevin examples above, if $C(\mu_t)$ represents the covariance of $\mu_t$, the mean-field dynamics drive $C(\mu_t)$ toward the covariance of the target distribution $\rho$ as $t\to\infty$. Adaptive samplers mimic this behavior using a single chain in the long-time limit. In this sense, adaptive MCMC can be interpreted as a finite-time, single-chain approximation of mean-field ensemble dynamics, with temporal averages replacing averages over particles.

\james{Notably, the ensemble methods considered here can incorporate adaptation. Given an ensemble $\bfX_k$, we update an adaptive parameter via
\begin{equation}\label{eq:one_system_adaptation}
    \begin{split}
        \Theta_k\ &=\ (1-a_k)\Theta_{k-1} + a_k\Pi_{\Gamma}^{\varepsilon,K_{\mathrm{cov}}}\left(\frac{1}{N}\sum_{i=1}^N(X_k^i-\widehat{\mu}_k)(X_k^i-\widehat{\mu}_k)^\top\right),\\
        \widehat\mu_k\ &=\ \frac{1}{N}\sum_{i=1}^NX_k^i,
    \end{split}
\end{equation}
where $N>2$, and $\{a_k\}_{k\geq 0}$ is a sequence such that $0<a_k < 1$ for all $k$, and there exists a $K>0$ such that $\{a_k\}_{k\geq K}$ is monotonically decreasing and satisfies
\begin{equation*}
    \sum_{k\geq K}a_k^2 <\infty,\qquad\sum_{k\geq K}a_k=\infty.
\end{equation*}
This effectively computes an average using all $N$ chains over the $K$ adaptation updates so far. The corresponding ensemble transition kernel
\begin{equation}
    \mathbf{P}_{\Theta_k}(\bfX_k,\,\widetilde \bfX_{k+1})=\prod_{i=1}^NP_{\Theta_k}(X^i_k,\, \widetilde X^i_{k+1})
\end{equation}
inherits geometric ergodicity from the per-particle kernel $P_{\Theta_k}$ whenever the latter is geometrically ergodic uniformly in $\Theta_k\in\Gamma$. Under additional regularity and containment assumptions, such schemes can be analyzed using standard adaptive-MCMC tools \citep{roberts2007, Atchad2006}. In this work, however, we use \emph{finite} adaptation (Remark~\ref{rmk:finite_adaptation}) and rely only on fixed-kernel convergence after the adaptation cutoff, sidestepping the need to verify diminishing adaptation in our experimental protocol.

\begin{remark}
    The classical adaptive MCMC scheme is recovered by taking adaptation weights of the form $a_k=\mathcal{O}(1/k)$, which decrease monotonically. We use the more general notation $a_k$ to allow for the restart scheme of Section~\ref{sec:step_size_randomization}. Under this scheme, the adaptation weights may be non-monotone before the final restart time $K$, but are chosen to decrease monotonically for all $k\geq K$.
\end{remark}

In a two-system adaptive variant, we can separate the adaptation for each subsystem. For example, we could maintain separate estimates for the $X$ and $Y$ groups:
\begin{equation}\label{eq:two_system_adaptation}
    \begin{split}
        \Theta_k^X\ &=\ (1-a_k)\Theta_{k-1}^X + a_k\Pi_{\Gamma}^{\varepsilon,K_{\mathrm{cov}}}\left(\frac{1}{N}\sum_{i=1}^N(X_k^i-\widehat{\mu}_k^X)(X_k^i-\widehat{\mu}_k^X)^\top\right),\\
        \Theta_k^Y\ &=\ (1-a_k)\Theta_{k-1}^Y + a_k\Pi_{\Gamma}^{\varepsilon,K_{\mathrm{cov}}}\left(\frac{1}{N}\sum_{i=1}^N(Y_k^i-\widehat{\mu}_k^Y)(Y_k^i-\widehat{\mu}_k^Y)^\top\right),\\
        \widehat{\mu}_k^X\ &=\ \frac{1}{N}\sum_{i=1}^NX_k^i,\qquad \widehat{\mu}_k^Y\ =\ \frac{1}{N}\sum_{i=1}^NY_k^i.
    \end{split}
\end{equation}
Then $\Theta_k^X$ is used when proposing updates for subsystem $Y$, and $\Theta_k^Y$ when updating subsystem $X$. In this way, each half of the ensemble adapts based on the other half, maintaining cross-system independence at each step.}

\daniel{The two adaptation modes \eqref{eq:one_system_adaptation} and \eqref{eq:two_system_adaptation} are made concrete in Section~\ref{sec:ensemble_chain_mcmc}, which presents Metropolis-adjusted two-system MALA and MAKLA-BCSS-2 algorithms in both the static-ensemble and adaptive forms (Algorithms~\ref{algo:interacting_mala}--\ref{algo:adaptive_makla}).}

\section{Algorithmic Realizations of Two-System Samplers}\label{sec:ensemble_chain_mcmc}
We now describe concrete two-system MCMC algorithms derived from Langevin dynamics. Our focus is on two specific instances: an overdamped sampler based on MALA \citep{parisi1981ula, besag1994mala}, and an underdamped sampler based on a new BCSS-2 instantiation of MAKLA, which we refer to as \emph{MAKLA-BCSS-2} to distinguish it from the original Verlet-based MAKLA of \citet{BouRabee2024}. We emphasize, however, that the two-system approach is general and can be applied to other MCMC proposal families as well (see \citet{clarte2022, sprungk2023} for alternative ensemble samplers).

In the ensemble-chain (multi-chain) setting, convergence to the target $\rho$ is guaranteed provided each Metropolis–Hastings update satisfies detailed balance. We therefore present the Coupled MALA and Coupled MAKLA-BCSS-2 samplers in Algorithms~\ref{algo:interacting_mala} and \ref{algo:interacting_makla}; at each iteration, the preconditioner is the empirical covariance of the opposite subsystem's current ensemble (no Robbins--Monro running average is maintained), so the kernel is memoryless conditional on the joint state $(\bfX,\bfY)$ and the two-system cross-coupling preserves $\rho^{\otimes 2N}$ exactly at finite $N$.

Adaptive samplers require additional care in practice, as their convergence depends on the target distribution's geometry, the adaptation scheme, and the proposal kernels. As a compromise, we present \emph{finite} adaptive samplers in the sense of the following remark.

\begin{remark}[Finite adaptation]\label{rmk:finite_adaptation}
We update the tuning parameter $\Theta_k$ only for $k\leq K_{\max}$, and then fix $\Theta=\Theta_{K_{\max}}$ for all subsequent iterations. For iterations $k\geq K_{\max}$, the algorithm reduces to a standard (non-adaptive) MCMC with kernel $P_{\Theta_{K_{\max}}}$, so classical convergence results apply. In particular, if $P_{\Theta_{K_{\max}}}$ is $\rho$-irreducible and aperiodic, the post-cutoff segment converges to the target and (after a short additional burn-in) can be treated as draws from a stationary chain.

Finite adaptive samplers also offer practical improvements: (i) the post-cutoff segment requires no per-step covariance update or Cholesky factorization, and is therefore much cheaper than the adaptive kernel; (ii) classical fixed-kernel diagnostics (e.g., $\widehat R$) apply directly to the post-cutoff samples; and (iii) it avoids the need to verify the diminishing-adaptation and geometric-ergodicity conditions as part of the experimental protocol.
\end{remark}

We present the finite adaptive two-system MALA and finite adaptive two-system MAKLA-BCSS-2 in Algorithms~\ref{algo:adaptive_mala} and \ref{algo:adaptive_makla}. The analogous one-system versions, which we also use in our experiments, are described below.

\subsection{Assumptions for Finite Adaptive MALA and \james{Finite Adaptive MAKLA-BCSS-2}}
\begin{assumption}\label{assumption:polynomial_form}
    The target can be written as $\rho(x) = h(x)e^{-p(x)}$ or $\rho(x) = h(x)^{-p(x)}$, with $h(x) \ge 0$ and $h, p$ polynomial.
\end{assumption}
The \emph{truncated} MALA kernel $P_C$ has proposal density
\begin{equation}\label{eq:truncated_mala_proposal}
    Q_C(X,Y)\ =\ \mathcal{N}\!\big(Y;\ X + h\, C\, D_\delta(X),\ 2hC\big),
    \quad
    D_\delta(X)\ =\ \frac{\delta}{\max(\delta,\,|\nabla\log\rho(X)|)} \nabla\log\rho(X),
\end{equation}
followed by the Metropolis--Hastings accept--reject step. The drift truncation by $\delta$ in $D_\delta(X)$ prevents proposals from being driven arbitrarily far when $|\nabla\log\rho(X)|$ is large. Under Assumption~\ref{assumption:polynomial_form}, the results of \citet{Atchad2006} show that the truncated MALA kernel with infinite adaptation satisfies the law of large numbers. However, our experimental validity relies only on the post-cutoff fixed-kernel phase (Remark~\ref{rmk:finite_adaptation}), and not a verbatim match between Atchad\'e's setting and our two-system ensemble construction.

{\begin{assumption}\label{assumption:c2_and_lipschitz}
    The target distribution satisfies $\log\rho(x)\in C^2$ and $\nabla \log\rho(x)$ is globally Lipschitz continuous.
\end{assumption}
Using a similar proof to \cite{durand2023}, one can prove that under Assumption~\ref{assumption:c2_and_lipschitz} the MAKLA-BCSS-2 described in Algorithms~\ref{algo:interacting_makla} and \ref{algo:adaptive_makla} is Lebesgue irreducible, aperiodic, and is reversible with respect to $\rho$.

We note that \citet{durand23a} recently proposed a \emph{finite} adaptive scheme based on the MALT sampler.\footnote{The three Metropolis-adjusted underdamped samplers use distinct integrator splittings: \emph{MALT} \citep{durand2023} wraps a Verlet (BAB) leapfrog inside an OBABO outer shell with full momentum refreshment; the original \emph{MAKLA} of \citet{BouRabee2024} uses an OABAO outer shell with persistent momentum (flipped on rejection); and our \emph{MAKLA-BCSS-2} keeps the OBABO outer shell with persistent momentum but replaces the Verlet inner core by the two-stage palindromic BABAB splitting of \citet{blanes2014} (BCSS-2; specified in Algorithm~\ref{algo:interacting_makla}). All numerical results in this paper that we label ``MAKLA'' family use this MAKLA-BCSS-2 instantiation.} However, a key distinction is that our adaptation explicitly preconditions the gradient $\nabla \log\rho(x)$ in the proposal step, whereas the adaptation in \citet{durand23a} preconditions the distribution used for momentum refreshment. By directly preconditioning the gradient updates (Section \ref{sec:ill_conditioned_distribution}), our algorithm aligns the dynamics with the local curvature of $\log\rho(x)$, accelerating exploration along stiff directions and mitigating random-walk behavior. This is particularly important in high dimensions, where poorly scaled gradients can severely slow convergence and reduce effective sample size.}

\daniel{Throughout Algorithms~\ref{algo:interacting_mala}--\ref{algo:adaptive_makla}, $C(\bfX)=\frac{1}{N-1}\sum_j(X^j-\bar{X})(X^j-\bar{X})^\top$ denotes the empirical (sample) covariance of an ensemble $\bfX=\{X^j\}_{j=1}^N$ (a positive semi-definite matrix), and $\Pi_{\Gamma}^{\varepsilon,K_{\mathrm{cov}}}$ is the cap-then-ridge projection of \eqref{eq:projection}, which maps any PSD input into $\Gamma$.}

\begin{algorithm}\caption{Coupled MALA}
\begin{algorithmic}\label{algo:interacting_mala}
\REQUIRE Desired distribution $\rho:\mathbb{R}^d\to[0,\infty)$, starting points $(\bfX_0^s)_j \in \mathbb{R}^{d}$ for $s=0,1$, and $j\in \{1,\ldots, N\}$, number of MCMC samples $M\geq 1$, step size $h>0$, covariance-norm cap $K_{\mathrm{cov}}>0$, and small $\varepsilon>0$ to ensure positive-definite covariance \\
\STATE{$s \gets 0$}
\FOR{$m\gets 1$, $m\leq M$}
    \STATE{$C^{1-s}_m\gets \Pi_{\Gamma}^{\varepsilon,K_{\mathrm{cov}}}(C(\bfX_{m-1}^{1 - s}))\qquad$ (Refresh covariance from the frozen subsystem $1-s$, whose state has not changed at iteration $m$)}
    \STATE{(Accept or reject particles individually)}
    \FOR{$j\in \{1, \ldots, N\}$ (in parallel)}
    \STATE{$(\bfX_m^s)_j \gets \text{ProposeMALA}_{h}((\bfX_{m-1}^s)_j, {C^{1-s}_m})$}
    \ENDFOR
    \STATE{$s\gets 1 - s\qquad$ (Switch systems)}
\ENDFOR
\STATE{}
\STATE{\textbf{Procedure} $\textrm{ProposeMALA}_{h}(X_0, C)$}
    \STATE{$\quad X_1\sim \mathcal{N}(X_0 + h C \nabla\log\rho(X_0), 2hC)\qquad$ (Propose a Langevin trajectory)}
    \STATE{$\quad \alpha \gets \textrm{AcceptRatioMALA}(X_0, X_1, C)$}
    \STATE{$\quad \beta \sim U(0, 1)$}
    \STATE{$\quad$\textbf{if} $\beta > \alpha$ \textbf{then}}
    \STATE{$\quad\quad$\textbf{return} $X_0$ $\qquad \textrm{(Reject)}$}
    \STATE{$\quad$\textbf{end if}}
    \STATE{$\quad$\textbf{return} $X_1$ $\qquad \textrm{(Accept)}$}
\STATE{\textbf{end Procedure}}
\STATE{}
\STATE{\textbf{Procedure} $\textrm{AcceptRatioMALA}(X, Y, C)$}
    \STATE{$\quad \alpha\gets \frac{\rho(Y)Q_{C}(X\mid Y)}{\rho(X)Q_C(Y\mid X)},\qquad\text{where}\ Q_C(y\mid x)\ \textrm{is the p.d.f. of $\mathcal{N}(x + hC\nabla\log\rho(x),2hC)$}$}
    \STATE{$\quad\textbf{return}\ \min(1,\alpha)$}
\STATE{\textbf{end Procedure}}
\end{algorithmic}
\end{algorithm}

\begin{algorithm}\small\caption{Coupled MAKLA-BCSS-2 (we use $L=1$ throughout this paper)}
\begin{algorithmic}\label{algo:interacting_makla}
\REQUIRE Desired distribution $\rho:\mathbb{R}^d\to[0,\infty)$, starting points $(\bfX_0^s)_j, (\bfV_0^s)_j \in \mathbb{R}^{d}$ for $s=0,1$, and $j\in \{1,\ldots, N\}$, number of MCMC samples $M\geq 1$, step size $h>0$, leap-frog steps $L>0$, persistence $\eta\in(0, 1)$, covariance-norm cap $K_{\mathrm{cov}}>0$, and small $\varepsilon>0$ to ensure positive-definite covariance  \\
\STATE{$s \gets 0$}
\FOR{$m\gets 1$, $m\leq M$}
    \STATE{$C^{1-s}_m\gets \Pi^{\varepsilon,K_{\mathrm{cov}}}_{\Gamma}(C(\bfX_{m-1}^{1 - s}))\qquad$ (Refresh covariance from the frozen subsystem $1-s$, whose state has not changed at iteration $m$)}
    \STATE{(Accept or reject particles individually)}
    \FOR{$j\in \{1, \ldots, N\}$ (in parallel)}
    \STATE{$((\bfX_m^s)_j, (\bfV_m^s)_j)\gets \text{ProposeMAKLA}^L_{h,\eta}((\bfX_{m-1}^s)_j, (\bfV_{m-1}^s)_j, {C^{1-s}_m})$}
    \ENDFOR
    \STATE{$s\gets 1 - s\qquad$ (Switch systems)}
\ENDFOR
\STATE{}
\STATE{\textbf{Procedure} $\textrm{ProposeMAKLA}^L_{h,\eta}(X_0, V_0, C)$}
    \STATE{$\quad (X_L, V_L, \Delta)\gets\textrm{OBABABO}^L_{h,\eta}(X_0, V_0, C)\qquad$ (Propose a Langevin trajectory)}
    \STATE{$\quad A \sim U(0, 1)$}
    \STATE{$\quad$\textbf{if} $A > \exp(-\Delta)$ \textbf{then}}
    \STATE{$\quad\quad$\textbf{return} $(X_0, -V_0)$ $\qquad \textrm{(Reject and flip momentum)}$}
    \STATE{$\quad$\textbf{end if}}
    \STATE{$\quad$\textbf{return} $(X_L, V_L)$ $\qquad \textrm{(Accept)}$}
\STATE{\textbf{end Procedure}}
\STATE{}
\STATE{\textbf{Procedure} $\textrm{OBABABO}^L_{h,\eta}(X_0, V_0,  C)$\quad \textit{(partial momentum refresh interleaved with the BCSS-2 BABAB Hamiltonian integrator of \citet{blanes2014})}}
    \STATE{$\quad$\textit{Inputs include the cached values $U_0=-\log\rho(X_0)$ and $\nabla U_0=-\nabla\log\rho(X_0)$ from the previous outer iteration; on the first call they are computed once.}}
    \STATE{$\quad b_1 \gets \tfrac{1}{6}(3-\sqrt{3}),\quad b_2 \gets 1-2b_1,\qquad a_1 \gets \tfrac12,\ a_2 \gets \tfrac12$ \quad \textit{(BCSS-2 palindromic coefficients)}}
    \STATE{$\quad \Delta\gets 0$}
    \STATE{$\quad$ \textbf{for } $i=0;\, i< L$ \textbf{do}}
    \STATE{$\quad\quad$ (O) $V \gets \sqrt{1-\eta}\, V_{i} + \sqrt{\eta}\,\xi_i^1$.\qquad \textit{(partial momentum refresh)}}
    \STATE{$\quad\quad$ \textit{Open the BCSS-2 BABAB Hamiltonian sub-step at the cached $(X_i,\nabla U_i)$:}}
    \STATE{$\quad\quad$ $V^{(0)}\gets V$;\quad $E_0 \gets U_i + \tfrac12\lVert V^{(0)}\rVert^2$ \quad \textit{(start-of-step BCSS-2 energy)}}
    \STATE{$\quad\quad$ (B$_1$) $V^{(1)} \gets V^{(0)} - b_1 h\,C^{1/2}\nabla U_i$. \quad \textit{(uses cached gradient at $X_i$)}}
    \STATE{$\quad\quad$ (A$_1$) $X^{(1)} \gets X_i + a_1 h\, C^{1/2} V^{(1)}$;\quad evaluate $\nabla U^{(1)}=-\nabla\log\rho(X^{(1)})$ \quad \textit{(fresh fused $(\log\{\rho\},\nabla\log\{\rho\})$ call)}.}
    \STATE{$\quad\quad$ (B$_2$) $V^{(2)} \gets V^{(1)} - b_2 h\,C^{1/2}\nabla U^{(1)}$.}
    \STATE{$\quad\quad$ (A$_2$) $X_{i+1} \gets X^{(1)} + a_2 h\,C^{1/2} V^{(2)}$;\quad evaluate $U_{i+1}=-\log\rho(X_{i+1})$, $\nabla U_{i+1}=-\nabla\log\rho(X_{i+1})$ \quad \textit{(cached for next iter)}.}
    \STATE{$\quad\quad$ (B$_1$) $V^{(3)} \gets V^{(2)} - b_1 h\,C^{1/2}\nabla U_{i+1}$.}
    \STATE{$\quad\quad$ $E_1\gets U_{i+1} + \tfrac12\lVert V^{(3)}\rVert^2$;\quad $\Delta \gets \Delta + (E_1 - E_0)$ \quad \textit{(deterministic BABAB energy change)}}
    \STATE{$\quad\quad$ (O) $V_{i+1} \gets \sqrt{1-\eta}\, V^{(3)} + \sqrt{\eta}\,\xi_i^2$.\qquad \textit{(partial momentum refresh)}}
    \STATE{$\quad$\textbf{end for}}
    \STATE{$\quad$\textbf{return} $(X_L, V_L, \Delta)$}
\STATE{\textbf{end Procedure}}
\end{algorithmic}
\end{algorithm}

\begin{algorithm}\caption{Two-System Finite-Adaptive Ensemble MALA}
\begin{algorithmic}\label{algo:adaptive_mala}
\REQUIRE Desired distribution $\rho:\mathbb{R}^d\to[0,\infty)$, starting points $(\bfX_0^s)_j \in \mathbb{R}^{d}$ for $j\in \{1,\ldots, N\}$ and $s=0,1$, number of MCMC samples $M\geq 1$, step size $h>0$, drift truncation $\delta > 0$, covariance-norm cap $K_{\mathrm{cov}} > 0$, small $\varepsilon > 0$, \james{total adaptation iterations $T_{\mathrm{adapt}}\ll M$}, and initial counter value $K_0 \geq 1$ in iterations (default $K_0 = \lceil\tau/(2h)\rceil$, with $\tau$ the first restart interval in diffusion-time units; see Section~\ref{sec:restart_init_counter}) \\
\STATE{$K\gets K_0$\qquad (Adaptive counter, initialized at $K_0$ rather than $1$)}
\FOR{$m\gets 1$, $m\leq M$}
    \STATE{\textit{(At outer iteration $m$, the inner $s$-loop updates subsystem $0$ first, then subsystem $1$. Each refresh uses the most recent state of the opposite subsystem: $\bfX^1_{m-1}$ when $s=0$, $\bfX^0_m$ when $s=1$.)}}
    \FOR{$s\in \{0,1\}$}
        \STATE{$\widetilde{\bfX}^{1-s}\gets\bfX^{1-s}_{m-1}$ if $s=0$, else $\bfX^{1-s}_{m}$ \qquad (most recently updated state of subsystem $1-s$)}
        \STATE{$\widetilde{C}^{1-s}\gets \Pi_{\Gamma}^{\varepsilon,K_{\mathrm{cov}}}(C(\widetilde{\bfX}^{1-s}))\qquad$ (Cap-then-ridge projection \eqref{eq:projection} of the empirical cov of the opposite subsystem)}
        \IF{$m\leq T_{\mathrm{adapt}}$}
        \STATE{$C^{1-s}_m\gets (1-1/K)\,C^{1-s}_{m-1} + (1/K)\,\widetilde{C}^{1-s}\qquad$ (Polyak update of the running cov)}
        \ELSE
        \STATE{$C^{1-s}_m\gets C^{1-s}_{m-1}\qquad$ (Adaptation frozen for $m>T_{\mathrm{adapt}}$)}
        \ENDIF
        \STATE{(Accept or reject particles)}
        \FOR{$j\in \{1, \ldots, N\}$ (in parallel)}
        \STATE{$(\bfX^s_m)_j\gets \text{ProposeTruncatedMALA}_{h,\delta}((\bfX_{m-1}^s)_j, {C^{1-s}_m})$}
        \ENDFOR
        \IF{$m\leq T_{\mathrm{adapt}}$}
        \STATE{$K\gets K + 1$}
        \ENDIF
    \ENDFOR
\ENDFOR
\STATE{}
\STATE{\textbf{Procedure} $\textrm{ProposeTruncatedMALA}_{h,\delta}(X, C)$}
    \STATE{$\quad d\gets \frac{\delta}{\max(\delta,|\nabla\log\rho(X)|)}\nabla \log\rho(X),$}
    \STATE{$\quad Y\sim \mathcal{N}(X + h C d, 2hC)\qquad$ (Propose a Langevin trajectory)}
    \STATE{$\quad \alpha \gets \textrm{AcceptRatioTruncatedMALA}(X, Y, C)$}
    \STATE{$\quad \beta \sim U(0, 1)$}
    \STATE{$\quad$\textbf{if} $\beta > \alpha$ \textbf{then}}
    \STATE{$\quad\quad$\textbf{return} $X$ $\qquad \textrm{(Reject)}$}
    \STATE{$\quad$\textbf{end if}}
    \STATE{$\quad$\textbf{return} $Y$ $\qquad \textrm{(Accept)}$}
\STATE{\textbf{end Procedure}}
\STATE{}
\STATE{\textbf{Procedure} $\textrm{AcceptRatioTruncatedMALA}_{h,\delta}(X, Y, C)$}
    \STATE{$\quad \alpha\gets \frac{\rho(Y)Q_{C}(X\mid Y)}{\rho(X)Q_C(Y\mid X)},\qquad\text{where}\ Q_C(y\mid x)\ \textrm{is the p.d.f. of $\mathcal{N}(x + hCD(x),2hC)$}$}
    \STATE{$\quad\textbf{return}\ \min(1,\alpha)$}
\STATE{\textbf{end Procedure}}
\end{algorithmic}
\end{algorithm}

\begin{algorithm}\caption{Two-System Finite-Adaptive Ensemble MAKLA-BCSS-2 (we use $L=1$ throughout this paper)}
\begin{algorithmic}\label{algo:adaptive_makla}
\REQUIRE Desired distribution $\rho:\mathbb{R}^d\to[0,\infty)$, starting points $(\bfX^s_0)_j, (\bfV^s_0)_j \in \mathbb{R}^{d}$ for $j\in \{1,\ldots, N\}$ and $s\in\{0,1\}$, number of MCMC samples $M\geq 1$, step size $h>0$, leap-frog steps $L>0$, persistence $\eta\in(0, 1)$, covariance-norm cap $K_{\mathrm{cov}} > 0$, small $\varepsilon > 0$, \james{total adaptation iterations $T_\mathrm{adapt}\ll M$}, and initial counter value $K_0 \geq 1$ in iterations (default $K_0 = \lceil\tau/(2h)\rceil$, with $\tau$ the first restart interval in diffusion-time units; see Section~\ref{sec:restart_init_counter}) \\
\STATE{$K\gets K_0$\qquad (Adaptive counter, initialized at $K_0$ rather than $1$)}
\FOR{$m\gets 1$, $m\leq M$}
    \STATE{\textit{(At outer iteration $m$, the inner $s$-loop updates subsystem $0$ first, then subsystem $1$. Each refresh uses the most recent state of the opposite subsystem: $\bfX^1_{m-1}$ when $s=0$, $\bfX^0_m$ when $s=1$.)}}
    \FOR{$s\in\{0, 1\}$}
        \STATE{$\widetilde{\bfX}^{1-s}\gets\bfX^{1-s}_{m-1}$ if $s=0$, else $\bfX^{1-s}_{m}$ \qquad (most recently updated state of subsystem $1-s$)}
        \STATE{$\widetilde{C}^{1-s}\gets \Pi_\Gamma^{\varepsilon,K_{\mathrm{cov}}} (C(\widetilde{\bfX}^{1-s}))\qquad$ (Cap-then-ridge projection \eqref{eq:projection} of the empirical cov of the opposite subsystem)}
        \IF{$m\leq T_{\mathrm{adapt}}$}
        \STATE{$C^{1-s}_m\gets (1-1/K)\,C^{1-s}_{m-1} + (1/K)\,\widetilde{C}^{1-s}\qquad$ (Polyak update of the running cov)}
        \ELSE
        \STATE{$C^{1-s}_m\gets C^{1-s}_{m-1}\qquad$ (Adaptation frozen for $m>T_{\mathrm{adapt}}$)}
        \ENDIF
        \STATE{(Accept or reject particles)}
        \FOR{$j\in \{1, \ldots, N\}$ (in parallel)}
        \STATE{$((\bfX^s_m)_j, (\bfV^s_m)_j)\gets \text{ProposeMAKLA}^L_{h,\eta}((\bfX^s_{m-1})_j, (\bfV^s_{m-1})_j, {C^{1-s}_m})$}
        \ENDFOR
        \IF{$m\leq T_{\mathrm{adapt}}$}
        \STATE{$K\gets K + 1$}
        \ENDIF
    \ENDFOR
\ENDFOR
\end{algorithmic}
\end{algorithm}

\daniel{\textbf{One-system variants.} The one-system Adaptive MALA / MAKLA-BCSS-2 samplers collapse the inner $s$-loop of Algorithms~\ref{algo:adaptive_mala}--\ref{algo:adaptive_makla} to a single $N$-chain ensemble $\bfX_m$ with one running covariance $C_m\in\Gamma$: (i) propose from $(X_{m-1})_j$ using the previous preconditioner $C_{m-1}$; (ii) project the post-step empirical covariance via \eqref{eq:projection}, $\widetilde{C}\leftarrow\Pi_{\Gamma}^{\varepsilon,K_{\mathrm{cov}}}(C(\bfX_m))$; (iii) Polyak-update $C_m\leftarrow(1-1/K)C_{m-1}+(1/K)\widetilde{C}$ (stays in $\Gamma$ by convexity); (iv) increment $K$. The cross-coupling between subsystems is replaced by self-coupling. During $m\leq T_{\mathrm{adapt}}$ this sacrifices exact finite-$N$ invariance for a single shared covariance estimate over all $N$ particles; for $m>T_{\mathrm{adapt}}$ the frozen kernel is a standard fixed-preconditioned MALA / MAKLA-BCSS-2 with $\rho$ as its per-chain invariant distribution by detailed balance. Empirically, the two designs are close, with the one-system variant marginally ahead on most \texttt{posteriordb} models and the two-system variant ahead on a few hard cases (Section~\ref{sec:simulations_experiments}).}

\daniel{The cap-then-ridge projection $\widetilde{C}^{1-s} \leftarrow \Pi_{\Gamma}^{\varepsilon,K_{\mathrm{cov}}}(C(\widetilde{\bfX}^{1-s})) = \varepsilon I_d + \alpha(C(\widetilde{\bfX}^{1-s}))C(\widetilde{\bfX}^{1-s})$ from \eqref{eq:projection} in Algorithms~\ref{algo:adaptive_mala}--\ref{algo:adaptive_makla} (and the analogous $\widetilde{C}$ in the one-system collapse) is not only motivated by theoretical requirements (a compact adaptation parameter set is needed for the convergence proofs of the adaptive-MCMC framework). It also improves numerical stability in practice: clipping the empirical-covariance component \emph{before} adding the ridge and \emph{before} the Polyak update protects the running estimate $C_m$ from rare outlier ensemble configurations, leaving the well-mixed history untouched. We have observed that this leads to better adaptation and noticeably more efficient samplers on some difficult problems with sharp curvature or ill-conditioned posteriors.}

\subsection{Step Size Control via Randomization}\label{sec:step_size_randomization}
The performance of Langevin-based methods is highly sensitive to the choice of step size $h$. Large values of $h$ reduce autocorrelation in the resulting Markov chain, but if $h$ is too large, proposals are frequently rejected, again leading to poor mixing. 

A rejection can be interpreted as evidence that the chain is locally sensitive to $\mathcal{O}(h)$ perturbations, and hence that the discretization error of the underlying SDE is too large. A classical remedy is delayed rejection \citep{Tierney1999, Mira2001, Modi_2024, turok2024}, in which a rejected proposal is followed by one or more subsequent proposals with progressively smaller step sizes. While theoretically attractive, this approach has two key drawbacks: (i) the computational cost of the $k^\text{th}$ proposal scales as $\mathcal{O}(2^k)$, and (ii) it is poorly suited to modern CPU/GPU hardware, where Single Instruction Multiple Data (SIMD) parallelism requires all chains to follow the same instruction path.

As an alternative, we introduce a randomized step size
\begin{equation}
    h = \gamma h_{\max}, \qquad 
    \gamma \sim \beta\, \delta_{x-1} + (1-\beta)\, f(x),
\end{equation}
where $f(x)$ is a distribution on $[0,1]$ with no mass at zero, and $0<\beta<1$, pictorially demonstrated in Fig. \ref{fig:step_size_distribution}. With probability $\beta$, the maximum step size $h_{\max}$ is used; otherwise, a random fraction $c \in (0,1)$ is drawn from $f(x)$, and $ch_{\max}$ is used. This randomized scheme preserves reversibility and can be applied independently or jointly across chains without disturbing invariance under $\rho^{\otimes 2N}$. Crucially, it automatically introduces occasional smaller step sizes when proposals are frequently rejected, thereby stabilizing the dynamics.

The choice of $f$ governs how often smaller steps are taken. A uniform distribution on $[0,1]$ is natural, but does not sufficiently emphasize small step sizes. We therefore adopt a distribution with density $f(x) = 3(1-x)^2$ on $[0,1]$, which biases the sampler toward smaller $h$ values while still allowing exploration with larger steps. \daniel{This density coincides with the law of $1 - U^{1/3}$ for $U\sim\text{Unif}(0,1)$, so a single uniform draw suffices to sample from $f$ at runtime.} The expected step size is then
\begin{equation}
    \mathbb{E}[h]
    = h_{\max}\int_0^1 x \Big( \beta \delta(x-1) + 3(1-\beta)(1-x)^2 \Big)\, dx
    = h_{\max}\left(\beta + \frac{1-\beta}{4}\right).
\end{equation}
\daniel{Throughout the experimental sections we set $\beta = 0.75$, so $\mathbb{E}[h] = 0.8125\,h_{\max}$.}

\daniel{This randomized step size scheme is particularly effective on targets with strong anisotropy that creates regions of large curvature demanding small steps, interspersed with regions of low curvature where larger steps are essential for efficient exploration.}

\daniel{Per-step jitter (rather than only after rejections) also protects against rare individual chains that wander into high-curvature tail regions where the global $h_{\max}$ is too large: occasional small steps fit the local curvature and free trapped chains, at negligible cost on the well-behaved bulk.}

\begin{figure}
    \centering
    \includegraphics[width=0.5\linewidth]{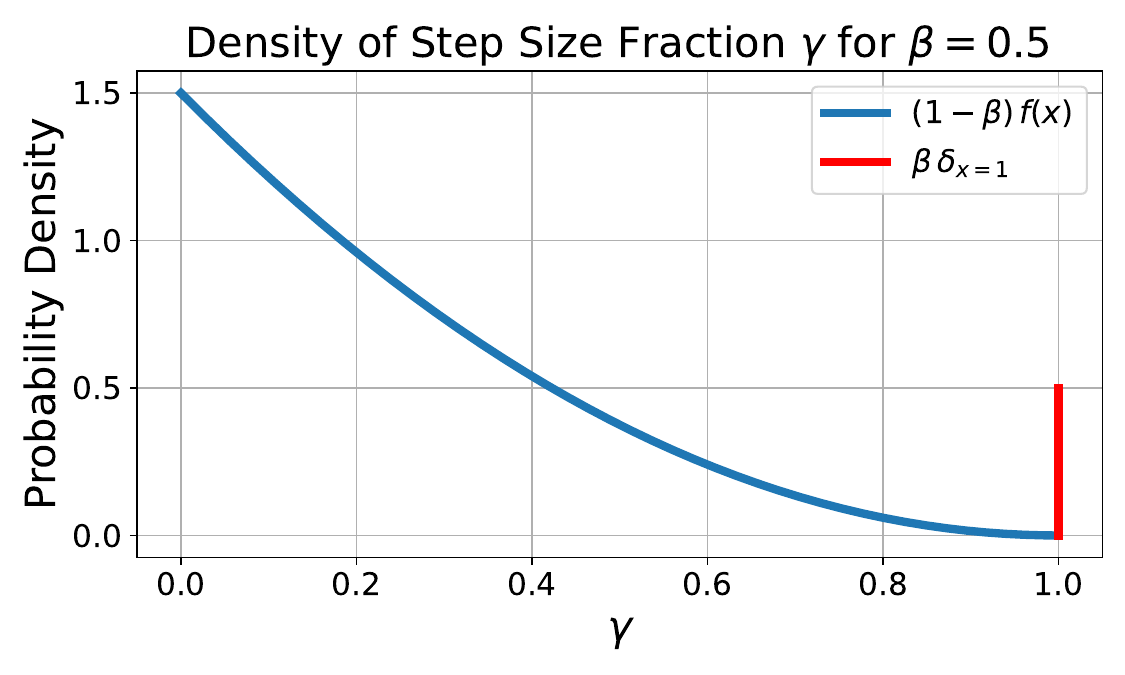}
    \caption{\footnotesize Visualization of the randomized step size distribution. The step size $ h = \gamma h_{\max} $ is drawn from a mixture of a point mass at $ \gamma = 1 $ with weight $ \beta \in (0,1) $, and a continuous component $ f(x) = 3(1-x)^2 $ supported on $ (0,1) $, with weight $ 1 - \beta $. This construction encourages frequent large proposals while allowing occasional small, exploratory steps, improving robustness across varying curvature scales.}
    \label{fig:step_size_distribution}
\end{figure}

\daniel{\textbf{Selecting $h_{\max}$ via a high-acceptance criterion.} The previous discussion specifies the law of the per-step jitter but leaves the maximum step size $h_{\max}$ free. We pick it by a geometric-ladder search: starting from a user-specified $h_0$, we successively shrink $h_{\max} \leftarrow 0.8\,h_{\max}$ until the empirical acceptance rate measured over a short trial run satisfies
\begin{equation}\label{eq:high_acc_criterion}
\mathrm{acc}(h_{\max}) \;\geq\; 1 - h_{\max}/c
\end{equation}
for a tuning constant $c$ (we use $c=16$ throughout the experiments). The criterion has a direct interpretation in diffusion-time units: with $c=16$, the per-step rejection probability is at most $h_{\max}/16$, so over a window of one unit of diffusion time -- which contains $\lceil 1/h_{\max}\rceil$ MAKLA iterations -- the expected number of rejected proposals is at most $1/16$. The chain therefore averages no more than one rejection per $16$ units of diffusion time, meaning the discrete MAKLA trajectory stays very close to a faithful discretization of the underlying kinetic Langevin diffusion (KLD): the Metropolis correction fires only sporadically, just often enough to repair accumulated integrator error and restore exact invariance under $\rho$. This is more aggressive than the constant-acceptance targets typical of MALA-type schemes (around $0.574$ in the optimal-scaling theory of \citet{RobertsRosenthal1998}) and is a deliberate design choice for the kinetic-Langevin family: in the underdamped regime, staying near the true SDE trajectory is the dominant route to fast mixing, and the rare-rejection regime is where that approximation is sharpest.}

\subsection{Restarting Adaptive Schemes}
The adaptive updates we consider in this paper are implemented via a running estimator
\begin{equation*}
\Theta_K \;=\; \left(1-\frac{1}{K}\right)\,\Theta_{K-1} \;+\; \frac{1}{K}\,\theta(X_K),
\end{equation*}
for some statistic $\theta(X_K)$ (covariance, step size, etc.). Since the iteration $K$ increases, later samples have diminishing influence on $\Theta_K$; yet the early samples (drawn far from stationarity) receive disproportionate weight. This mismatch slows the convergence of $\Theta_K$ to $\mathbb{E}_\rho[\theta(X)]$ and creates a negative feedback loop: poor $\Theta_K\Rightarrow$ poor proposals $\Rightarrow$ slow mixing $\Rightarrow$ poor $\Theta_K$. With ensembles sharing a common $\Theta_K$, the problem is exacerbated by sensitivity to initialization.

\daniel{To mitigate this, during the adaptation phase of length $T_{\mathrm{adapt}}$ we restart the counter $K$ on a fixed schedule $\mathcal{R} = \{\tau,\,2\tau,\,\dots,\,\tau_{\max}\}$ with $\tau_{\max} < T_{\mathrm{adapt}}$. At each $t\in\mathcal{R}$, we apply a \emph{partial reset}
$
K \leftarrow \rho\, K,\ \rho \in [0,1],
$
\emph{leaving $\Theta_K$ unchanged.} The rate $\rho$ controls how aggressively the past is forgotten: $\rho=1$ is a no-op; $\rho=0$ zeros the counter so the next observation completely overwrites $\Theta$ via $\gamma = 1$; $\rho=0.5$ (our default) halves $K$ so the next several Robbins--Monro updates receive $\gamma = 1/(\rho K + 1)$ -- roughly twice the size of the un-reset weight $\gamma = 1/(K+1)$ -- and recent samples therefore dominate $\Theta$ much faster than under the natural $\mathcal{O}(1/K)$ decay alone. The role of the restart is precisely this: by reweighting future samples upward, the disproportionate influence of the early, badly mixed samples on $\Theta_K$ is forgotten more quickly than the unmodified Robbins--Monro update permits on its own.}

\daniel{No restarts are performed after $\tau_{\max}$; the running estimator $\Theta_K$ then continues to evolve un-restarted over the remaining $T_{\mathrm{adapt}}-\tau_{\max}$ time units, so the preconditioner that is finally frozen at the end of the adaptation phase (Remark~\ref{rmk:finite_adaptation}) is dominated by samples drawn from a well-mixed chain. In our \texttt{posteriordb} experiments we take an adaptation phase of $T_{\mathrm{adapt}}=5{\,}000$ diffusion-time units, with $\tau_{\max} = 0.5\,T_{\mathrm{adapt}} = 2{\,}500$, ten partial restarts spaced uniformly at $\tau = 0.05\,T_{\mathrm{adapt}} = 250$ diffusion-time units (i.e.\ at $t \in \{250,500,750,\ldots,2500\}$), and $\rho = 0.5$. (The synthetic-target study in Section~\ref{sec:sysnthetic_experiments} uses the shorter $T_{\mathrm{adapt}}=2{\,}000$ with five restarts at $\tau = 0.1\,T_{\mathrm{adapt}} = 200$, since those well-conditioned targets equilibrate within a few hundred diffusion-time units; the larger \texttt{posteriordb} budget was set to give stiff posteriors -- particularly the ODE-based \texttt{one\_\allowbreak comp\_\allowbreak mm\_\allowbreak elim\_\allowbreak abs} -- enough time for the rolling cov to stabilize before being frozen.)}

\daniel{In practice these restarts lead to better-calibrated preconditioners, higher acceptance rates, and improved ESS/Grad during the sampling phase. Since adaptation is stopped at $T_{\mathrm{adapt}}$, the post-adaptation chain is a standard fixed-kernel MCMC sampler as in Remark~\ref{rmk:finite_adaptation}.}

\paragraph{Initializing the Robbins--Monro counter.}\label{sec:restart_init_counter}\daniel{We initialize the counter at $K_0 = \lceil\tau/(2h)\rceil$ iterations -- equivalent to one-half of the first restart interval $\tau$ measured in diffusion-time units -- rather than $K_0 = 1$, so that the initial preconditioner $\Theta_{\mathrm{init}}$ is not overwritten by a single noisy sample on the first iteration. This ``lukewarm'' start is small relative to the cumulative weight by the first restart at $t=\tau$, so the asymptotic Polyak behavior is unchanged; on the synthetic targets, it halves the early-iteration transient compared with $K_0=1$.}

\subsection{Dealing with Ill-Conditioned Distributions}\label{sec:ill_conditioned_distribution}
Posterior distributions arising in practice (e.g., from \texttt{posteriordb} \citep{magnusson2024posteriordb}) are often severely ill-conditioned. For example, in our experiments, the Hessian at the maximum a posteriori at times has eigenvalues with a dynamic range as large as $10^{11}$. This creates an extremely anisotropic gradient field, where the log-density is very steep in some directions and nearly flat in others. 

\textbf{Momentum Refresh:} For samplers that rely on integrating a chain of gradient evaluations within a single proposal (such as MALT, HMC, and MAKLA), multiple integrator steps can easily lead to numerical instability, often producing explosive trajectories with infinite or \texttt{NaN} values. Stability can be restored by dramatically shrinking the step size, but this in turn reduces sample efficiency. A more robust strategy is to \emph{use only a single integrator step} (Verlet leapfrog for MALT, BCSS-2 for our MAKLA-BCSS-2), which permits a larger step size without sacrificing stability. We have adopted this strategy ($L=1$) in all of our experiments.

However, in this regime, the way momentum is refreshed becomes crucial. HMC and MALT employ full-velocity resampling at each proposal, which reduces efficiency by discarding information about the current momentum. For HMC, this resampling is essential to guarantee convergence to the invariant distribution (via the Virial theorem \citep{Clausius1870, Duane1987}). In contrast, MALT and MAKLA approximate underdamped Langevin dynamics, in which momentum refreshment is already built in. Moreover, given any initial position and velocity, the continuous-time process converges to the correct invariant distribution, and the Metropolis correction merely removes discretization bias. Hence, additional momentum resampling is unnecessary, and retaining momentum (with sign reversal on rejection) yields more efficient sampling. For this reason, we adopt MAKLA as the base sampler.\\

\james{
\textbf{Rescaling the Distribution:} Ill-conditioning also implies that different coordinates of the state space live on vastly different scales. The smallest scales, associated with the largest curvature directions, require an excessively small step size $h$ to have a good acceptance rate. To mitigate this, we apply an affine change of variables $x=Az$, for some $A\in S^{d}_{++}$. The induced density for $z$ is
\begin{equation*}
    g(z)\ \propto\ |\det(A)|\,\rho(Az), \qquad \text{for some}\ A\in S^{d}_{++},
\end{equation*}
Since $A$ is fixed, $|\det A|$ is a constant and may be absorbed into the normalising constant; hence, we may sample using the unnormalized density $g(z)$ in place of $\rho(x)$. Given draws $z_i\sim g(z)$, samples from the original target are recovered by the inverse transform $x_i=Az_i$. This reparameterization balances coordinate scales, enables larger stable steps, and reduces numerical instabilities.

In our experiments, we choose $A$ based on the local curvature at the mode $x^*$, which is numerically determined. Specifically, we compute the negative log-density Hessian $H = -\nabla^2 \log\rho(x^*)$ and set
\begin{equation*}
A \;=\; (H + \varepsilon I_d)^{-1/2}, \qquad \text{for small}\qquad \varepsilon>0,
\end{equation*}
so that $A$ is the inverse symmetric square root of the local precision matrix (with $\varepsilon$ providing numerical regularization).

This construction can be viewed as a simple geometric preconditioning strategy. For Langevin or Hamiltonian samplers, working in the $z$-coordinates corresponds to evolving in a geometry that approximately whitens the local quadratic approximation of $-\log \{\rho\}$: it reduces anisotropy, partially decorrelates dimensions, and stabilizes the deterministic drift by improving effective smoothness/conditioning. 

Under standard regularity conditions in Bayesian asymptotics, the posterior becomes approximately Gaussian as the data size grows (Bernstein--von Mises theorem \citep{Vaart1998}), in which case $H^{-1}$ is a reliable global description of posterior covariance, and the fixed rescaling above can yield substantial gains. For strongly non-Gaussian targets (e.g., skewed, heavy-tailed, or multimodal distributions), the Hessian at $x^*$ may be poorly representative away from the mode, and a fixed inverse-Hessian rescaling offers limited improvement. In such regimes, mean-field or adaptive methods provide additional flexibility by learning global scale and correlation structure on the fly, thereby adapting the proposal geometry as the chain explores regions far from $x^*$.}

\section{Experiments}\label{sec:simulations_experiments}
We conducted a two-part empirical study. In Section \ref{sec:sysnthetic_experiments}, we evaluate MAKLA-BCSS-2 and MALA-family samplers on two controlled synthetic targets with closed-form moments -- a heavy-tailed Student-$t$ in $d=10$ and the curvature-aware Girolami--Calderhead banana in $d=2$ \citep{GirolamiCalderhead2011} -- to compare h-tuning criteria and isolate the regime in which each variant is reliable. In Section \ref{sec:posteriordb}, we move to real-data experiments, benchmarking the most reliable subset (the four MAKLA-BCSS-2 variants with the $\mathrm{acc}\ge 1-h/16$ criterion) on 45 posterior distributions from \texttt{posteriordb} \citep{magnusson2024posteriordb} against a broad set of NUTS variants.

For the adaptive methods, we consider two adaptation modes: (i) \textbf{1sys-adaptive}, which uses the full ensemble of $N$ chains to update a single shared covariance (as in Eqn. \eqref{eq:one_system_adaptation}), and (ii) \textbf{2sys-adaptive}, which splits the ensemble $N$ into two subsystems and adapts each half separately (as in Eqn. \eqref{eq:two_system_adaptation}) to combine long-term adaptation with cross-chain independence. We refer to the variants of Algorithms~\ref{algo:interacting_mala}--\ref{algo:interacting_makla} (which use the opposite subsystem's empirical covariance directly, with no Robbins--Monro averaging) as \emph{Coupled} MALA / MAKLA-BCSS-2 throughout.

\subsection{Synthetic Experiments}\label{sec:sysnthetic_experiments}
We evaluate the MAKLA-BCSS-2 and MALA samplers on two synthetic targets with closed-form moments, chosen to expose complementary failure modes:
\begin{itemize}
\item \textbf{Student-$t$} ($d=10$, $\nu=4$): $\rho(x) \propto \big(1+\tfrac{1}{\nu}x^\top A x\big)^{-(d+\nu)/2}$, with $A$ a fixed SPD matrix whose eigenvalues are linearly spaced in $[10^{-2},10^{2}]$. Heavy tails and anisotropic but globally smooth geometry; the static Hessian-at-MAP preconditioner is exact in the Gaussian limit.
\item \textbf{Banana} ($d=2$): in the form used by \citet{GirolamiCalderhead2011} as the canonical RMHMC benchmark,
\begin{equation*}
\rho(y_1,y_2)\;\propto\;\exp\!\left[-\tfrac{y_1^2}{2\sigma_1^2}-\tfrac{1}{2}\left\{y_2-b(y_1^2-\sigma_1^2)\right\}^2\right],\qquad \sigma_1=10,\ b=0.1.
\end{equation*}
The bend $b(y_1^2-\sigma_1^2)$ produces strong curvature that the Hessian-at-MAP cannot capture (Figure~\ref{fig:banana_density}).
\end{itemize}

\begin{figure}[!ht]
\centering
\includegraphics[width=0.55\linewidth]{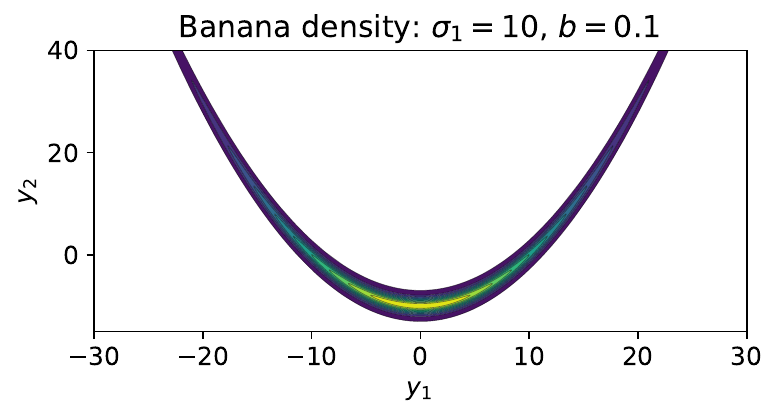}
\caption{\footnotesize Density contours of the Girolami--Calderhead banana \citep{GirolamiCalderhead2011} with $\sigma_1=10$, $b=0.1$, used as a two-dimensional benchmark for curved, non-Gaussian geometry.}
\label{fig:banana_density}
\end{figure}

\textit{Experimental setup.} The bench protocol matches Section~\ref{sec:posteriordb}: $140$ parallel chains across $14$ cores, Hessian-at-MAP preconditioning into rescaled coordinates $Y$, and a diffusion-time budget of $2{\,}000$ for burn-in and $8{\,}000$ for sampling (with a separate $2{\,}000$-unit adaptation phase upstream of the burn-in for the two adaptive MAKLA-BCSS-2 / MALA variants, exactly as in Section~\ref{sec:posteriordb}). The MAKLA-BCSS-2 family uses the BCSS-2 BABAB inner integrator with the high-acceptance ``decay'' criterion $\mathrm{acc}\ge 1-h/16$ of Section~\ref{sec:step_size_randomization} and a geometric ladder $h_{k+1}=0.8 h_k$ starting from $h_0 = 2.4$. The MALA family uses the Roberts--Rosenthal fixed-acceptance criterion $\mathrm{acc}\ge \alpha^*=0.574$ \citep{RobertsRosenthal1998} starting from $h_0=1.0$. The adaptive variants of both families use the truncation caps from Algorithms~\ref{algo:adaptive_mala}--\ref{algo:adaptive_makla}: $K_{\mathrm{cov}} = 10{\,}000$ for the operator-norm cap on the running covariance, and additionally $\delta = 10{\,}000$ for the drift truncation in adaptive MALA. Both act as numerical-stability safety nets. We report only the MAKLA-BCSS-2 decay configuration and the MALA fixed-acc configuration: in pilot experiments, MAKLA-BCSS-2 + fixed-acc rejected too aggressively on the banana ($\widehat{R}_{\max}\ge 1.06$ on every variant), and MALA + decay shrank $h$ far below the MALA-optimal $\mathrm{acc}\approx 0.574$ regime, producing $3$--$1{\,}200\times$ more gradient evaluations per ESS on these targets. The chosen pairings are therefore the natural regime for each family.

\textit{Results.} Table~\ref{tab:synthetic_grad_per_ess} reports the maximum-component gradient evaluations per effective sample, $\max_j(\mathrm{Grad}/\mathrm{ESS}_j)$ (lower is better), with bootstrap standard errors over $200$ resamples of chains. Cells flagged with \dag\ correspond to runs whose chains did not reach $\widehat{R}_{\max}\le 1.05$.

\begin{table}[!ht]
\centering
\caption{\footnotesize Worst-component gradient evaluations per effective sample on the two synthetic targets, with bootstrap standard errors. The MAKLA-BCSS-2 family uses the BCSS-2 BABAB inner integrator (two gradient evaluations per outer step, counted in the Grad totals); the MALA family uses one gradient evaluation per step. \textbf{Bold}: best converged value in column. \dag: chains did not reach $\widehat{R}_{\max}\le 1.05$.}
\label{tab:synthetic_grad_per_ess}
\footnotesize
\setlength{\tabcolsep}{6pt}
\begin{tabular}{l|cc|cc}
\toprule
& \multicolumn{2}{c|}{MAKLA-BCSS-2, decay ($h_0{=}2.4$)} & \multicolumn{2}{c}{MALA, fixed-acc ($h_0{=}1$)} \\
Variant         & banana                                & student-$t$              & banana                              & student-$t$    \\
\midrule
basic           & $2.01\times 10^4 \pm 6{\,}907\,\dag$  & $\phantom{0}4.8\pm 0.6$            & $1.66\times 10^5 \pm 1.6\times 10^4\,\dag$ & $37\pm 5$ \\
1sys-adaptive   & $\mathbf{1{\,}539 \pm 158}$           & $\mathbf{\phantom{0}4.5\pm 0.7}$   & $2.52\times 10^4 \pm 6{\,}062$              & $\mathbf{19\pm 3}$ \\
2sys-adaptive   & $1{\,}996 \pm 328$                    & $\phantom{0}5.3\pm 1.8$            & $\mathbf{2.25\times 10^4 \pm 3{\,}687}$     & $21\pm 4$ \\
coupled         & $2{\,}826 \pm 968$                    & $15.4\pm 5.8$                      & $2.36\times 10^4 \pm 5{\,}954$              & $35\pm 3$ \\
\bottomrule
\end{tabular}
\end{table}

\textit{Two findings.}
\begin{enumerate}
\item \emph{Adaptation is essential where the static Hessian misses the local curvature.} On the banana, basic MAKLA-BCSS-2 fails to converge ($\widehat R_{\max} = 1.03$) at $\sim\!2\times 10^4$ Grad/ESS\textsubscript{worst}; the three in-sampler-adapting MAKLA-BCSS-2 variants converge at $\sim\!1{\,}500$--$2{\,}800$, an order of magnitude better. On Student-$t$ the static Hessian is essentially exact, so basic MAKLA-BCSS-2 ($4.8$) and the 1-/2-sys adaptive variants ($4.5$--$5.3$) are tied; only Coupled MAKLA-BCSS-2 degrades ($15.4$) because its $14\cdot 8d \approx 1100$-particle ensemble is overkill for $d=10$. Adaptation buys an order of magnitude exactly where the global Hessian-at-MAP linearization is locally inaccurate, and nothing where it is accurate.\vspace{-0.6em}
\item \emph{Underdamped beats overdamped.} The MALA fixed-acc variants are $4$--$15\times$ behind their MAKLA-BCSS-2 counterparts on both targets ($19$--$37$ on Student-$t$, $2.3$--$2.5\times 10^4$ on banana), the expected gap between first-order Metropolized drift-diffusion and second-order BCSS-2 underdamped integration.
\end{enumerate}

We therefore benchmark the four MAKLA-BCSS-2 variants (all tuned with the decay criterion $\mathrm{acc}\ge 1-h/16$) against NUTS on $45$ \texttt{posteriordb} posteriors in Section~\ref{sec:posteriordb}.

\subsection{\texttt{posteriordb} Benchmarks: MAKLA-BCSS-2 family vs.\ NUTS variants}
\label{sec:posteriordb}

We benchmarked the MAKLA-BCSS-2 family of samplers against a broad set of No-U-Turn Sampler (NUTS) variants on 45 posterior distributions from \texttt{posteriordb}.
\texttt{posteriordb} contained 47 examples with reference samples; for two of them (\texttt{eight\_\allowbreak schools-\allowbreak eight\_\allowbreak schools\_\allowbreak centered} and \texttt{mcycle\_gp-accel\_gp}), the posterior mode could not be found with any optimizer, and the gradient norm kept increasing during optimization (potentially due to an improper posterior), so we omitted those two from the comparison.

\paragraph{Samplers compared.} We evaluate twelve samplers: four MAKLA-BCSS-2 variants in \texttt{BlackJAX} (static-Hessian, 1- and 2-system Adaptive, Coupled), five \texttt{BlackJAX} NUTS variants (with/without Hessian preconditioning, dual-averaging or window adaptation, diagonal or full mass), and three \texttt{CmdStan} NUTS variants (Hess+DA with the metric frozen at the inverse Hessian at MAP, plus the \texttt{Stan}-default WA diagonal- and full-mass settings without preconditioning).\footnote{The \texttt{CmdStan} family is restricted to three variants because \texttt{Stan}'s window adaptation overwrites any user-supplied initial metric, so the analogues of \texttt{CmdStan} \emph{NUTS, Hessian preconditioned, WA (diag/full mass)} would coincide with the corresponding non-preconditioned variants.} Table~\ref{tab:geometric_means} lists the abbreviated names used throughout.

\paragraph{Experimental settings.} The MAKLA-BCSS-2, 1-system Adaptive, and 2-system Adaptive samplers used $140$ parallel chains during burn-in and sampling (split evenly between subsystems for the two-system variant). Coupled MAKLA-BCSS-2 used $8d$ particles. All MAKLA-BCSS-2 samplers used the velocity-refreshment parameter $\eta=\exp(-\gamma h)$ with $\gamma=0.1$, $L=1$ outer step per proposal with the BCSS-2 BABAB inner core of \citet{blanes2014},\footnote{With cached gradient and log-posterior values across consecutive outer iterations, this costs two $\nabla\log\rho$ and two $\log\rho$ evaluations per step. Switching the inner core to Verlet (BAB) halves this cost, but BCSS-2's larger admissible step size more than compensates; see Section~\ref{sec:ensemble_chain_mcmc}.} and the high-acceptance step-size criterion $\mathrm{acc}\ge 1-h/16$ of Section~\ref{sec:step_size_randomization}.

\daniel{For the two Adaptive MAKLA-BCSS-2 variants, burn-in and sampling are preceded by a dedicated adaptation phase constructing the frozen ensemble covariance $\widehat\Sigma$ used as the sample-phase preconditioner. \textit{Adaptation:} a $20$-chain ensemble (split $10{+}10$ for the two-system variant; $\widehat\Sigma=\tfrac12(\widehat\Sigma^0+\widehat\Sigma^1)$) runs for $T_{\mathrm{adapt}} = 5{\,}000$ diffusion-time units of Algorithm~\ref{algo:adaptive_makla}, with $K_0 = \lceil\tau/(2h)\rceil$ (Section~\ref{sec:restart_init_counter}), ten partial restarts at $\mathcal{R}=\{250,\ldots,2500\}$ ($\rho=0.5$), and operator-norm cap $K_{\mathrm{cov}}=10{\,}000$; the adapt step size comes from a geometric-ladder pre-search ($h_{k+1}=0.8 h_k$, $h_0=2.4$) at the largest $h$ meeting $\mathrm{acc}(h)\ge 1-h/16$. \textit{Refinement:} we absorb $\widehat{\Sigma}$ into $Z=Y\,L^{-T}$ ($L=\mathrm{chol}(\widehat{\Sigma})$) and rerun the ladder against frozen $\widehat\Sigma$ on $140$ chains, walking \emph{up} from $h_{\mathrm{adapt}}$. \textit{Burn-in/sampling:} $140$ chains are bootstrap-resampled from the $20$ adapt-final positions, then run for $5{\,}000\cdot\lceil 1/h^\star\rceil$ burn-in and $30{\,}000\cdot\lceil 1/h^\star\rceil$ sample steps ($2{\,}000/8{\,}000$ for Coupled). Bootstrap initialization seeds chains where adaptation saw $\mathrm{acc}\approx 1-h^\star/16$, avoiding from-MAP starts that hurt ill-conditioned posteriors. Static MAKLA-BCSS-2 skips adaptation/refinement: its $140$ chains start at $\mathrm{MAP}+1$-stdev Gaussian noise in Hessian-rescaled coordinates with $h$ chosen by a single geometric-ladder pass.}

All NUTS variants (both \texttt{BlackJAX} and \texttt{CmdStan}) used $140$ parallel chains, $2\,000$ warm-up iterations, and $8\,000$ sampling iterations, with target acceptance $0.8$ and a tree-depth cap of $10$ doublings.

The mode was found by particle-swarm optimization followed by trust-region Newton-CG from \texttt{scipy}, and the static Hessian preconditioner was computed at this mode. For the single ODE-based posterior (\texttt{one\_\allowbreak comp\_\allowbreak mm\_\allowbreak elim\_\allowbreak abs}), the \texttt{JAX} models use \\ \texttt{jax.experimental.ode.odeint} (Dormand--Prince 5(4)) with relative and absolute tolerances $10^{-12}$, and the Stan model uses \texttt{ode\_bdf} (implicit BDF) with default tolerances ($10^{-10}/10^{-10}/10^{8}$ steps); the tighter \texttt{JAX} setting was needed to suppress mild $\widehat R$ drift on this stiff posterior.

For non-preconditioned NUTS variants, we initialize the $140$ sample chains at the warm-up's final position rather than reset to MAP. All NUTS variants (preconditioned and not) use the \texttt{BlackJAX} default tree-depth cap of $10$ doublings.

We report \emph{worst-case efficiency}: the minimum-across-components effective sample size per gradient evaluation, equivalent to the maximum gradient evaluations per effective sample, denoted $\max_j (\text{Grad}/\text{ESS}_j)$. Lower is better.

\paragraph{ESS estimator and bootstrap SEs.} We treat the $M=140$ chains as approximately independent samplers (a well-justified assumption once burn-in completes and $\widehat R < 1.01$, which we verify in Table~\ref{tab:rhat_failures}). For each component $j$, with per-chain online sample means $\widehat\mu_m^{(j)}$ and biased sample variances $\widehat{s}_m^{(j)2}$, define $\widehat{V}_{\mathrm{within},j} = \tfrac{1}{M}\sum_m \widehat{s}_m^{(j)2}$ and $\widehat{V}_{\mathrm{between},j} = \mathrm{Var}_{m}(\widehat{\mu}_m^{(j)})$ (computed from the law of total variance to avoid catastrophic cancellation when $|\bar\mu_j|^2 \gg \mathrm{Var}_j$). Per-chain effective sample size is then
\[
\widehat{\mathrm{ESS}}_j^{(\mathrm{per-chain})} \;=\; \frac{\widehat{V}_{\mathrm{within},j}+\widehat V_{\mathrm{between},j}}{\widehat{V}_{\mathrm{between},j}},
\]
i.e.\ posterior variance over chain-mean variance. This is the standard i.i.d.-equivalent ESS, justified by treating the $M$ chain means as $M$ independent estimators of the posterior mean. Total ESS is $M$ times the per-chain value; Grad/ESS uses the per-chain total gradient count (one grad/step for Verlet/MALA, two for BCSS-2, $\langle$tree-depth$\rangle \cdot 2$ for NUTS), making it apples-to-apples across integrators. To quantify Monte Carlo uncertainty, we report bootstrap standard errors: $B=200$ resamples of the $M$ per-chain $(\widehat{\mu}_m,\widehat{s}_m^2)$ pairs with replacement, with the bootstrap SD across replicates. Every $X\pm s$ entry in this paper has $X$ from the direct all-$M$-chain estimator and $s$ from this bootstrap; the figure bars are $\pm 1$ bootstrap SD around the direct point estimate.

\paragraph{Efficiency versus dimension (Figure \ref{fig:gradess-posteriordb}).} We plot the maximum gradient evaluations per effective sample (the worst component) on a logarithmic scale, with one panel per sampler. Each panel includes a horizontal dashed line for that sampler's geometric mean across all models.

\hspace*{2em} The four MAKLA-BCSS-2 family methods cluster tightly between $1$ and $10$ for nearly all 45 models, with geometric means of $\max_j(\mathrm{Grad}/\mathrm{ESS}_j)$ in the $1.88$--$2.61$ range under the BCSS-2 inner integrator with the $5{\,}000$-dt adaptation schedule. The Hessian-preconditioned NUTS variants (in both \texttt{BlackJAX} and \texttt{CmdStan}) are competitive on the bulk of the panel, with geometric means between $4.32$ and $6.24$. The non-preconditioned NUTS variants degrade further: $9.64$--$9.99$ for the full-mass settings and $71.06$--$72.76$ for the diagonal-mass settings, with a heavier upper tail on the harder posteriors.

\paragraph{Per-gradient efficiency (Table \ref{tab:geometric_means}).}
The four MAKLA-BCSS-2 variants cluster tightly between $1.88$ and $2.61$ Grad/ESS\textsubscript{worst} -- 2-system Adaptive MAKLA-BCSS-2 leads at $\mathbf{1.88}$ -- while the eight NUTS variants span $4.32$ (\texttt{CmdStan} Hess + DA, the best NUTS) to $72.76$ (\texttt{CmdStan} non-preconditioned, diagonal mass). The best MAKLA-BCSS-2 variant therefore outperforms every NUTS variant by at least $\sim\!2.3\times$ on the geometric mean. Within NUTS, the dominant factor is whether the Hessian is folded in: \texttt{BlackJAX}'s Hessian-preconditioned variants reach $5.43$--$6.24$, while removing Hessian preconditioning costs an extra $1.6$--$13\times$. The two implementations of \emph{NUTS, Hess + DA} agree to within $\sim\!25\%$ ($4.32$ \texttt{CmdStan} vs.\ $5.43$ \texttt{BlackJAX}).

\paragraph{Wall-clock throughput (Table \ref{tab:geometric_means} and Figure \ref{fig:ess-per-sec-posteriordb}).} The four MAKLA-BCSS-2 variants tie at $1.70$--$1.83\!\times\!10^6$ ESS\textsubscript{worst}/sec, with 2-system Adaptive MAKLA-BCSS-2 leading at $\mathbf{1.83\!\times\!10^6}$. The best NUTS variant reaches $3.12\!\times\!10^5$, so the MAKLA-BCSS-2 family delivers a $\sim\!5.5\times$--$5.9\times$ throughput advantage on top of the per-gradient gains. The last two columns of Table~\ref{tab:geometric_means} report the across-model median step size $\widetilde h$ and per-iteration gradient cost $\widetilde N_{\nabla}$: MAKLA-BCSS-2 runs at $\widetilde h\approx 1.5$--$1.9$ with $\widetilde N_{\nabla}=2$, against $\widetilde h\approx 0.10$--$0.98$ and $\widetilde N_{\nabla}\approx 4$--$34$ for the NUTS variants.

\begin{table}[!ht]
    \centering
    \caption{\footnotesize Geometric means across all 45 \texttt{posteriordb} models of the maximum gradient evaluations per effective sample $\max_j (\text{Grad}/\text{ESS}_j)$ (lower is better) and of the total worst-component effective samples per second of sampling wall-clock time summed across all parallel chains (higher is better). The last two columns report the across-model median step size $\widetilde h$ and median per-iteration gradient cost $\widetilde N_{\nabla}$ (for NUTS, $\widetilde N_{\nabla}$ equals the median number of leapfrog steps per NUTS iteration; for MAKLA-BCSS-2, each outer step costs two BCSS-2 fresh gradients). \textbf{Bold} indicates the best value in the first two columns across all twelve samplers. Abbreviations used in the sampler column (and reused in Tables~\ref{tab:rhat_failures}--\ref{tab:hard_models}): \emph{BkJ} = BlackJAX, \emph{CS} = CmdStan, \emph{Hess} = Hessian-preconditioned, \emph{DA} = dual-averaging step adaptation (metric frozen), \emph{WA} = window adaptation, \emph{diag/full} = diagonal or full-dense mass-matrix adaptation.}
    \footnotesize
    \renewcommand{\arraystretch}{0.95}
    \begin{tabular}{l|cccc}
         \toprule
         Sampler & $\max_j (\text{Grad}/\text{ESS}_j)$ & Total $\text{ESS}_{\min}/\text{sec}$ & $\widetilde h$ & $\widetilde N_{\nabla}$\\
         \midrule
         MAKLA-BCSS-2 (static Hessian)      &$1.93$              & $1.78\times 10^{6}$         & $1.92$ & $2.0$ \\
         1-sys Adaptive MAKLA-BCSS-2        &$1.96$              & $1.77\times 10^{6}$         & $1.92$ & $2.0$ \\
         2-sys Adaptive MAKLA-BCSS-2        &$\mathbf{1.88}$     & $\mathbf{1.83\times 10^{6}}$ & $1.92$ & $2.0$ \\
         Coupled MAKLA-BCSS-2               &$2.61$              & $1.70\times 10^{6}$         & $1.54$ & $2.0$ \\
         \midrule
         BkJ NUTS, Hess + DA                 & $5.43$              & $2.68\times 10^{5}$        & $0.98$ & $4.0$  \\
         BkJ NUTS, Hess + WA (diag)          & $6.18$              & $3.12\times 10^{5}$        & $0.71$ & $5.9$  \\
         BkJ NUTS, Hess + WA (full)          & $6.24$              & $2.91\times 10^{5}$        & $0.71$ & $5.9$  \\
         BkJ NUTS, WA (diag)                 & $71.06$             & $2.64\times 10^{4}$        & $0.10$ & $34.3$ \\
         BkJ NUTS, WA (full)                 & $9.99$              & $1.86\times 10^{5}$        & $0.69$ & $6.1$  \\
         \midrule
         CS NUTS, Hess + DA                 & $4.32$              & $1.35\times 10^{5}$         & $0.96$ & $3.7$  \\
         CS NUTS, WA (diag, STAN default)   & $72.76$             & $9.98\times 10^{3}$         & $0.11$ & $32.2$ \\
         CS NUTS, WA (full)                 & $9.64$              & $6.61\times 10^{4}$         & $0.67$ & $6.6$  \\
        \bottomrule
    \end{tabular}
    \label{tab:geometric_means}
\end{table}

\paragraph{Posterior-mean accuracy.} Taking the 1-system Adaptive MAKLA-BCSS-2 estimate ($140$ chains $\times\,30{\,}000/h$ samples) as reference, all twelve samplers agree on coordinate-wise posterior means to within $0.04$ standard deviations on every model ($0.02$ within the MAKLA-BCSS-2 family), and posterior standard deviations agree to within $1\%$. Differences across samplers therefore manifest only as efficiency, not as bias.

\paragraph{Convergence diagnostics (Table \ref{tab:rhat_failures}).} We adopt the conventional Gelman--Rubin threshold $\widehat R < 1.01$ \citep{Gelman2013}. All four MAKLA-BCSS-2 variants converge on every one of the 45 models ($\widehat R_{\max} \le 1.0078$). Among the eight NUTS variants, three pass on all 45 models (BkJ Hess + WA diag, BkJ Hess + WA full, and CS WA diag); the remaining five fail on at most $2$ models, with \texttt{one\_\allowbreak comp\_\allowbreak mm\_\allowbreak elim\_\allowbreak abs} (a stiff ODE posterior) accounting for almost every failure (full per-sampler list in Table~\ref{tab:rhat_failures}).

\begin{table}[!ht]
    \centering
    \caption{\footnotesize Convergence: number of models where the classical Gelman--Rubin diagnostic $\widehat{R}_{\max}$ across all parameters exceeded $1.01$, and the maximum value of $\widehat{R}_{\max}$ across all 45 models. See Table~\ref{tab:geometric_means} for sampler-name abbreviations.}
    \footnotesize
    \renewcommand{\arraystretch}{0.95}
    \begin{tabular}{l|cc|p{6cm}}
         \toprule
         Sampler & \# failed & Max $\widehat{R}_{\max}$ & Failing models ($\widehat{R}_{\max}$)\\
         \midrule
         MAKLA-BCSS-2 (static Hessian)      &$0$ & $1.0078$ & --- \\
         1-sys Adaptive MAKLA-BCSS-2        &$0$ & $1.0057$ & --- \\
         2-sys Adaptive MAKLA-BCSS-2        &$0$ & $1.0021$ & --- \\
         Coupled MAKLA-BCSS-2               &$0$ & $1.0019$ & --- \\
         \midrule
         BkJ NUTS, Hess + DA                 & $2$ & $1.0412$ & \texttt{one\_comp} ($1.041$); \texttt{gp\_pois\_regr} ($1.023$) \\
         BkJ NUTS, Hess + WA (diag)          & $0$ & $1.0096$ & --- \\
         BkJ NUTS, Hess + WA (full)          & $0$ & $1.0070$ & --- \\
         BkJ NUTS, WA (diag)                 & $1$ & $1.0102$ & \texttt{one\_comp} ($1.010$) \\
         BkJ NUTS, WA (full)                 & $1$ & $1.0835$ & \texttt{one\_comp} ($1.084$) \\
         \midrule
         CS NUTS, Hess + DA                 & $1$ & $1.0217$ & \texttt{one\_comp} ($1.022$) \\
         CS NUTS, WA (diag)                 & $0$ & $1.0032$ & --- \\
         CS NUTS, WA (full)                 & $1$ & $1.0267$ & \texttt{one\_comp} ($1.027$) \\
        \bottomrule
    \end{tabular}
    \label{tab:rhat_failures}
\end{table}

\paragraph{Where adaptation and ensemble averaging matter most: three hard models (Table \ref{tab:hard_models}).}
Most of the 45 models are well-conditioned regressions that the static Hessian-at-MAP isotropises adequately, but three non-Gaussian posteriors tell a more revealing story: the funnel \texttt{eight\_\allowbreak schools\_\allowbreak noncentered} ($d{=}10$), the stiff GP-Poisson \texttt{gp\_pois\_regr} ($d{=}13$), and the ODE-based pharmacokinetic \texttt{one\_\allowbreak comp\_\allowbreak mm\_\allowbreak elim\_\allowbreak abs} ($d{=}4$). On all three, in-sampler adaptation or coupling reduces static MAKLA-BCSS-2's Grad/ESS by $\sim 1.7$--$6.5\times$ (Table~\ref{tab:hard_models}); the best adaptive/coupled MAKLA-BCSS-2 variant is $1.3$--$4.6\times$ ahead of the best converged NUTS variant on each. The $5{\,}000$-dt adaptation schedule was motivated by \texttt{one\_\allowbreak comp\_\allowbreak mm\_\allowbreak elim\_\allowbreak abs}: at the shorter $2{\,}000$-dt budget the adapt chains under-mix and $\widehat\Sigma$ misses key anisotropy directions, breaking the bootstrap-init guarantee.

\begin{table}[!ht]
    \centering
    \caption{\footnotesize Worst-component gradient evaluations per effective sample on the three hardest models in our benchmark, with bootstrap standard errors over $200$ resamples. Lower is better. Top block: MAKLA-BCSS-2 family. Middle block: \texttt{BlackJAX} NUTS variants. Bottom block: \texttt{CmdStan} NUTS variants. Cells flagged with \dag\ correspond to runs whose chains did not reach $\widehat{R}_{\max}\le 1.01$. See Table~\ref{tab:geometric_means} for sampler-name abbreviations.}
    \footnotesize
    \renewcommand{\arraystretch}{0.95}
    \begin{tabular}{l|ccc}
        \toprule
        Sampler                            & \texttt{gp\_pois\_regr}& \texttt{eight\_schools\_nc}& \texttt{one\_\allowbreak comp\_\allowbreak mm\_\allowbreak elim\_\allowbreak abs}\\
        & $(d=13)$ & $(d=10)$ & $(d=4)$\\
        \midrule
        MAKLA-BCSS-2 (static Hessian)      &$794.7 \pm 91.4$               & $13.12 \pm 1.12$         & $721.7 \pm 175.9$\\
        1-sys Adaptive MAKLA-BCSS-2        &$127.2 \pm 14.0$               & $\phantom{0}8.48 \pm 1.06$ & $773.7 \pm 264.7$\\
        2-sys Adaptive MAKLA-BCSS-2        &$\mathbf{122.5 \pm 14.6}$      & $\mathbf{7.91 \pm 1.01}$ & $314.7 \pm 58.1$\\
        Coupled MAKLA-BCSS-2               &$202.3 \pm 9.9$                & $\phantom{0}9.44 \pm 0.56$ & $\mathbf{189.7 \pm 17.1}$\\
        \midrule
        BkJ NUTS, Hess + DA                 & $14{\,}125 \pm 9{\,}800\,\dag$  & $35.13 \pm 3.62$         & $3{\,}315 \pm 1{\,}595\,\dag$\\
        BkJ NUTS, Hess + WA (diag)          & $1{\,}373 \pm 151.2$            & $28.03 \pm 3.24$         & $983.8 \pm 253.3$\\
        BkJ NUTS, Hess + WA (full)          & $2{\,}725 \pm 940.2$            & $\phantom{0}9.99 \pm 1.25$ & $906.1 \pm 170.5$\\
        BkJ NUTS, WA (diag)                 & $543.2 \pm 56.1$               & $13.65 \pm 1.77$         & $1{\,}386 \pm 279.7\,\dag$\\
        BkJ NUTS, WA (full)                 & $2{\,}231 \pm 1{\,}453$          & $15.54 \pm 1.38$         & $8{\,}212 \pm 2{\,}002\,\dag$\\
        \midrule
        CS NUTS, Hess + DA                 & $5{\,}026 \pm 1{\,}526$           & $37.45 \pm 3.92$         & $1{\,}732 \pm 493.1\,\dag$\\
        CS NUTS, WA (diag)                 & $4{\,}154 \pm 3{\,}357$           & $13.71 \pm 1.60$         & $425.0 \pm 55.4$\\
        CS NUTS, WA (full)                 & $1{\,}782 \pm 444.7$             & $11.86 \pm 1.34$         & $2{\,}727 \pm 685.0\,\dag$\\
        \bottomrule
    \end{tabular}
    \label{tab:hard_models}
\end{table}

\paragraph{Interpretation.} All four MAKLA-BCSS-2 variants beat every NUTS variant on both geometric-mean Grad/ESS and throughput, with 2-sys Adaptive leading on both axes; the static-vs-adaptive MAKLA-BCSS-2 gap is small in aggregate ($\sim 1\%$) but widens to $1.5$--$4\times$ on the hard models of Table~\ref{tab:hard_models} where the static Hessian linearisation is locally inaccurate. Within NUTS, the dominant lever is Hessian preconditioning ($1.6$--$13\times$ in geometric-mean Grad/ESS): \texttt{BlackJAX} can absorb it (reaching $5.43$--$6.24$ even with windowed mass) while \texttt{CmdStan} cannot, since Phase II overwrites any user-supplied metric, leaving its non-preconditioned variants at $9.64$--$72.76$.

\begin{figure}[!ht] 
    \centering
    \includegraphics[width=0.99\linewidth]{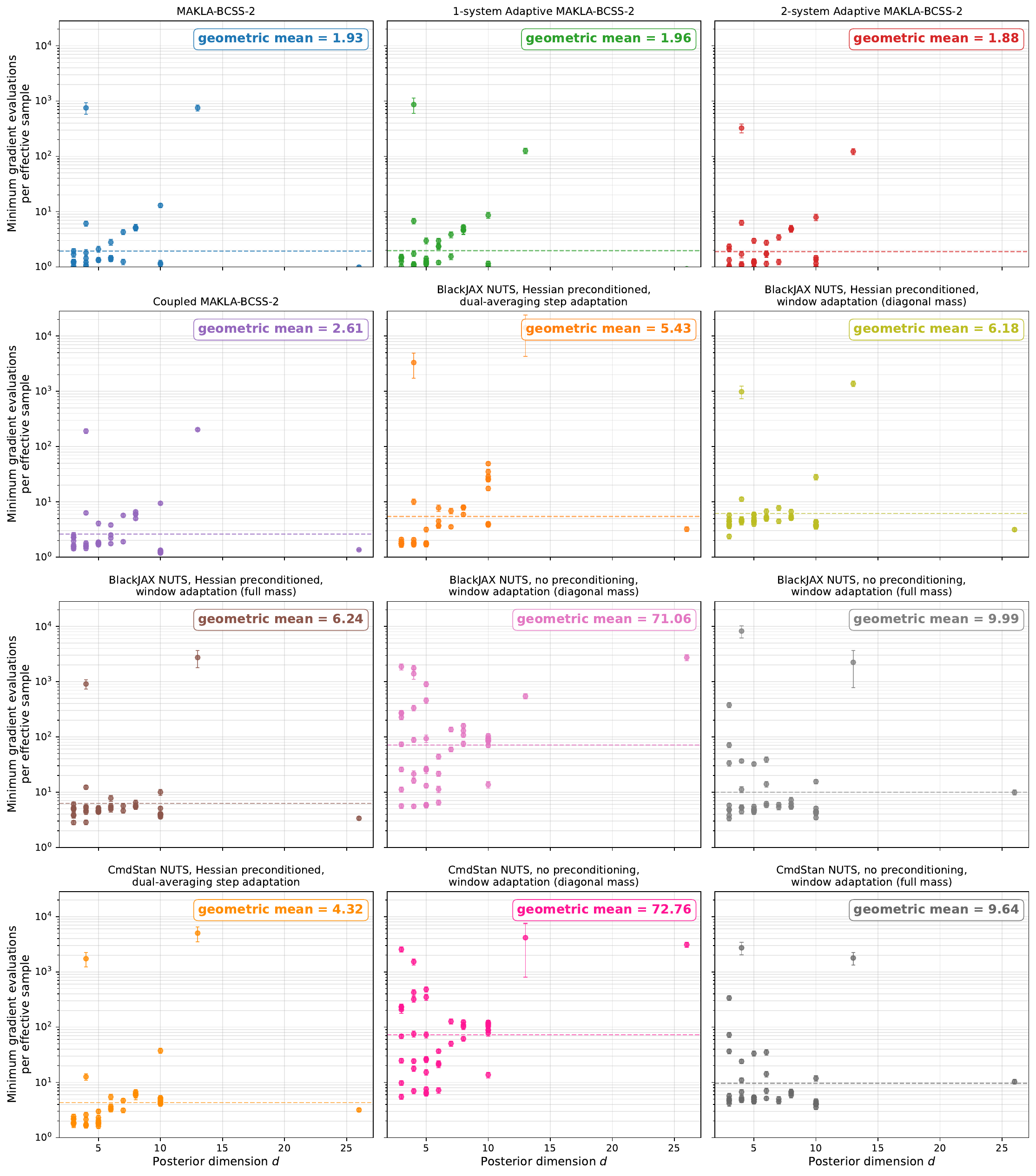}
  \caption{\footnotesize \textbf{Maximum gradient evaluations per effective sample (taken across all parameter components, i.e.\ the worst-mixing component) as a function of posterior dimension, across 45 \texttt{posteriordb} benchmark posteriors, for the twelve samplers compared.} Each dot is one posterior; the $y$-axis is on a logarithmic scale and lower values are better. The horizontal dashed line in each panel is that sampler's geometric mean across all 45 models, with the numerical value annotated in the upper right of the panel. Vertical bars indicate $\pm 1$ bootstrap standard error over $200$ chain resamples. Top row: MAKLA-BCSS-2 family. Middle two rows: \texttt{BlackJAX} NUTS variants. Bottom row: \texttt{CmdStan} NUTS variants.}
  \label{fig:gradess-posteriordb}
\end{figure}

\begin{figure}[!ht]
    \centering
    \includegraphics[width=0.99\linewidth]{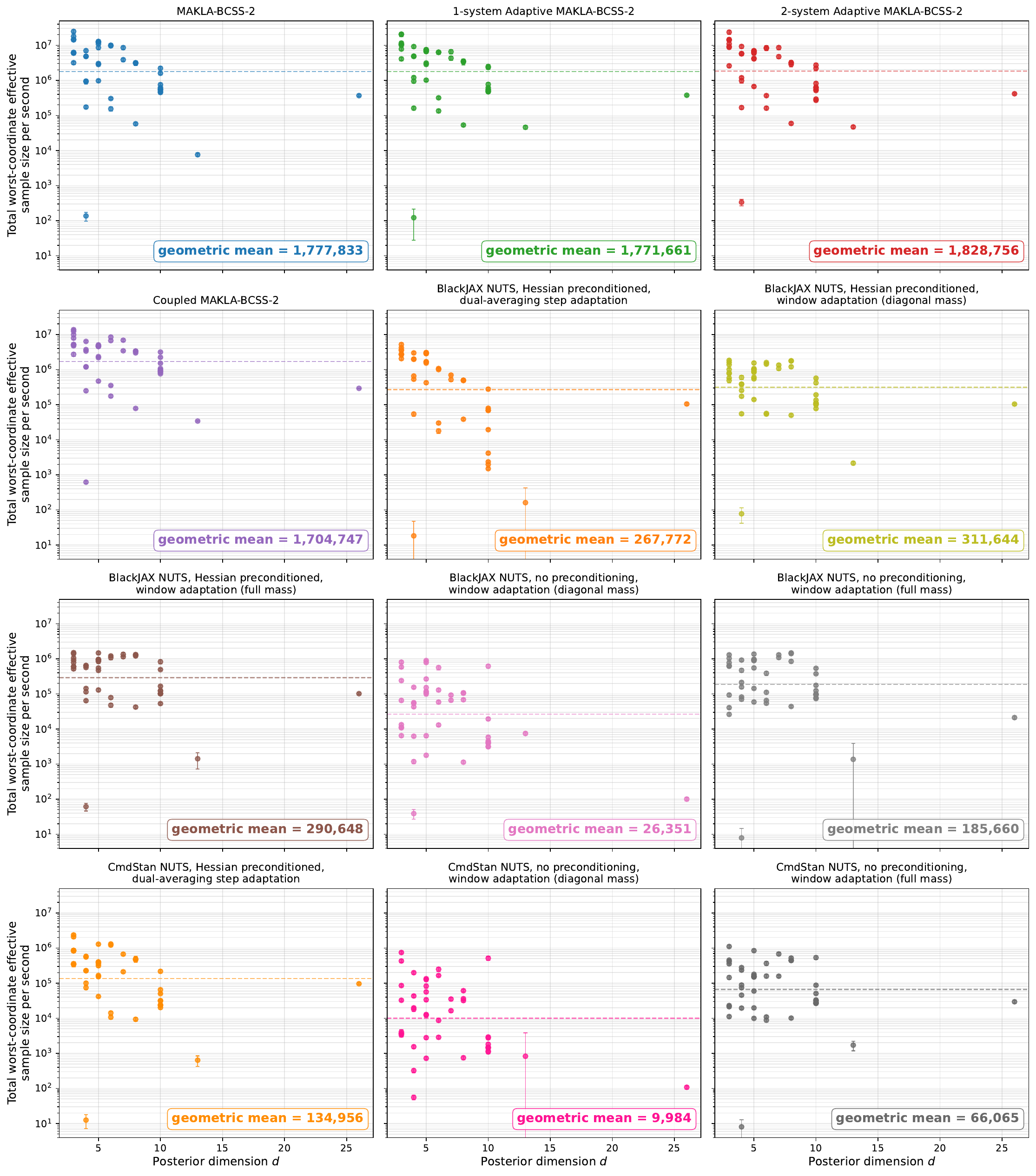}
  \caption{\footnotesize \textbf{Total worst-component effective samples per second of sampling wall-clock time (summed across all $140$ parallel chains) as a function of posterior dimension, across 45 \texttt{posteriordb} benchmark posteriors, for the twelve samplers compared.} Each dot is one posterior; the $y$-axis is on a logarithmic scale and higher values are better. The horizontal dashed line in each panel is that sampler's geometric mean across all 45 models, with the numerical value annotated in the lower right of the panel. Vertical bars indicate $\pm 1$ bootstrap standard error over $200$ chain resamples.}
  \label{fig:ess-per-sec-posteriordb}
\end{figure}

\subsection{Higher-dimensional models}\label{sec:higherdim}

The 45 \texttt{posteriordb} models in Section~\ref{sec:posteriordb} all have $d \le 100$. To probe how the MAKLA-BCSS-2 family scales we evaluated three additional \texttt{posteriordb} models that are \emph{genuinely} high-dimensional, in the sense that no conjugacy relationship exists for the latent structure and no analytic marginalization reduces them to a lower-dimensional problem: a Bayesian generalized linear mixed model \texttt{GLMM\_data\_GLMM1\_model} ($d=237$ unconstrained), a closed-population mark--recapture model \texttt{Mth\_data\_Mth\_model} ($d=394$), and the IRT 2PL latent-regression model \texttt{fims\_Aus\_Jpn\_irt\_2pl\_latent\_reg\_irt} ($d=531$). All three have non-Gaussian, dimension-scaling latent components, so the entire latent vector must be sampled jointly with the population-level parameters.

We compare the eight JAX-based samplers from Section~\ref{sec:posteriordb} but omit Coupled MAKLA-BCSS-2: its $14\cdot 8d$-particle ensemble ($\sim\!4{\,}000$--$8{\,}500$ particles for these models) exceeds available memory. The 1-sys and 2-sys Adaptive MAKLA-BCSS-2 variants instead build $\widehat\Sigma$ from a $20$-chain ensemble via Robbins--Monro averaging, so the per-step covariance need not be full-rank.

Tables~\ref{tab:hd_grad_per_ess}--\ref{tab:hd_ess_per_sec} report worst-component gradient evaluations per effective sample and total worst-component effective samples per second of sampling wall-clock time. All cells use the \texttt{BlackJAX} default tree-depth cap of $10$ doublings. \emph{All eight samplers converge cleanly} ($\widehat R_{\max} \le 1.002$) on $\texttt{GLMM1}$ and $\texttt{Mth}$. On $\texttt{fims\_2pl\_irt}$, all seven adaptive samplers converge ($\widehat R_{\max} \le 1.0004$); only static-Hessian MAKLA-BCSS-2 fails ($\widehat R_{\max} = 1.10$), where the local Hessian-at-MAP is too poor to keep the chains in the typical set without in-sampler adaptation. The dagger $\dag$ marks this single failed cell in the tables below; ESS\textsubscript{min}/sec is reported as $0$ for it. Posterior means agree to within $0.013$ standard deviations across all converged samplers (worst case: GPU Adaptive MAKLA-BCSS-2 on \texttt{Mth}; CPU agreement is within $0.007$), and posterior standard deviations agree to within $1\%$.

\begin{table}[!ht]
    \centering
    \caption{\footnotesize Worst-component gradient evaluations per effective sample on the three higher-dimensional \texttt{posteriordb} models, with bootstrap standard errors over $200$ chain resamples. Lower is better. \textbf{Bold}: best converged value in column. $\dag$: chains did not reach $\widehat R_{\max}\le 1.01$.}
    \footnotesize 
    \begin{tabular}{l|ccc}
        \toprule
        Sampler                            & \texttt{GLMM1}     & \texttt{Mth}              & \texttt{fims\_2pl\_irt} \\
                                           & $(d=237)$          & $(d=394)$                 & $(d=531)$ \\
        \midrule
        MAKLA-BCSS-2 (static Hessian)      &$5.83 \pm 0.42$    & $544.6 \pm 62.9$          & $26{\,}470 \pm 1{\,}871\,\dag$ \\
        1-sys Adaptive MAKLA-BCSS-2        &$\mathbf{4.72 \pm 0.39}$ & $\mathbf{82.98 \pm 8.23}$ & $16.41 \pm 1.67$ \\
        2-sys Adaptive MAKLA-BCSS-2        &$4.90 \pm 0.36$    & $127.2 \pm 13.7$          & $\mathbf{14.81 \pm 1.53}$ \\
        \midrule
        BkJ NUTS, Hess + DA                & $113.9 \pm 8.1$    & $1{\,}566 \pm 192$         & $3{\,}636 \pm 376$ \\
        BkJ NUTS, Hess + WA (diag)         & $9.65 \pm 0.66$    & $248.8 \pm 24.2$           & $747.4 \pm 80.8$ \\
        BkJ NUTS, Hess + WA (full)         & $12.81 \pm 0.96$   & $\mathbf{75.30 \pm 7.29}$  & $2{\,}395 \pm 222$ \\
        BkJ NUTS, WA (diag)                & $11.24 \pm 1.00$   & $116.8 \pm 11.8$           & $279.5 \pm 32.5$ \\
        BkJ NUTS, WA (full)                & $12.98 \pm 1.08$   & $89.3 \pm 12.7$            & $52.9 \pm 5.4$ \\
        \bottomrule
    \end{tabular}
    \label{tab:hd_grad_per_ess}
\end{table}

\begin{table}[!ht]
    \centering
    \caption{\footnotesize Total worst-component effective samples per second of sampling wall-clock time, summed across all $140$ parallel chains (and $20$ for Adaptive MAKLA-BCSS-2's adapt phase, which is amortized separately). Higher is better. \textbf{Bold}: best converged value in column. $\dag$: chains did not reach $\widehat R_{\max}\le 1.01$.}
    \footnotesize 
    \begin{tabular}{l|ccc}
        \toprule
        Sampler                            & \texttt{GLMM1}             & \texttt{Mth}            & \texttt{fims\_2pl\_irt} \\
                                           & $(d=237)$                  & $(d=394)$               & $(d=531)$ \\
        \midrule
        MAKLA-BCSS-2 (static Hessian)      &$68{\,}499 \pm 4{\,}847$     & $194 \pm 24$            & $1.15 \pm 0.08\,\dag$ \\
        1-sys Adaptive MAKLA-BCSS-2        &$\mathbf{85{\,}952 \pm 6{\,}707}$ & $\mathbf{1{\,}251 \pm 126}$ & $1{\,}696 \pm 179$ \\
        2-sys Adaptive MAKLA-BCSS-2        &$77{\,}353 \pm 5{\,}310$     & $817 \pm 92$            & $\mathbf{1{\,}955 \pm 198}$ \\
        \midrule
        BkJ NUTS, Hess + DA                & $2{\,}772 \pm 156$          & $73.6 \pm 9.7$           & $6.6 \pm 0.7$ \\
        BkJ NUTS, Hess + WA (diag)         & $47{\,}201 \pm 2{\,}755$     & $583.6 \pm 58.4$         & $50.3 \pm 5.4$ \\
        BkJ NUTS, Hess + WA (full)         & $19{\,}868 \pm 1{\,}224$     & $1{\,}170 \pm 123$        & $9.0 \pm 0.8$ \\
        BkJ NUTS, WA (diag)                & $54{\,}124 \pm 4{\,}053$     & $\mathbf{1{\,}637 \pm 174}$ & $171.3 \pm 20.7$ \\
        BkJ NUTS, WA (full)                & $26{\,}518 \pm 1{\,}764$     & $792.9 \pm 113.7$         & $470.7 \pm 46.1$ \\
        \bottomrule
    \end{tabular}
    \label{tab:hd_ess_per_sec}
\end{table}

\begin{table}[!ht]
    \centering
    \caption{\footnotesize Tuned step size $h^\star$ and per-iteration gradient cost $N_{\nabla}$ on the three high-dimensional models. For NUTS, $N_{\nabla}$ is the average leapfrog steps per iteration (one gradient per step); for MAKLA-BCSS-2, $N_{\nabla}=2$ per outer step. Coupled MAKLA-BCSS-2 omitted (memory; see Section~\ref{sec:higherdim}). $\dag$: did not reach $\widehat R_{\max}\le 1.01$.}
    \footnotesize 
    \renewcommand{\arraystretch}{0.9}
    \setlength{\tabcolsep}{3pt}
    \begin{tabular}{l|cc|cc|cc}
        \toprule
        & \multicolumn{2}{c|}{\texttt{GLMM1} ($d{=}237$)} & \multicolumn{2}{c|}{\texttt{Mth} ($d{=}394$)} & \multicolumn{2}{c}{\texttt{fims\_2pl\_irt} ($d{=}531$)}\\
        Sampler & $h^\star$ & $N_{\nabla}$ & $h^\star$ & $N_{\nabla}$ & $h^\star$ & $N_{\nabla}$\\
        \midrule
        MAKLA-BCSS-2 (static Hessian)      & $0.79$ & $2.0$  & $0.21$ & $2.0$  & $0.40\,\dag$ & $2.0\,\dag$ \\
        1-sys Adaptive MAKLA-BCSS-2        & $0.79$ & $2.0$  & $0.63$ & $2.0$  & $0.98$       & $2.0$       \\
        2-sys Adaptive MAKLA-BCSS-2        & $0.79$ & $2.0$  & $0.50$ & $2.0$  & $0.98$       & $2.0$       \\
        \midrule
        BkJ NUTS, Hess + DA                 & $0.42$ & $81.0$ & $0.17$ & $200.7$ & $0.24$      & $655.3$    \\
        BkJ NUTS, Hess + WA (diag)          & $0.37$ & $15.0$ & $0.16$ & $31.0$  & $0.017$     & $255.0$    \\
        BkJ NUTS, Hess + WA (full)          & $0.37$ & $15.0$ & $0.33$ & $15.0$  & $0.24$      & $1{\,}009.8$ \\
        BkJ NUTS, WA (diag)                 & $0.38$ & $15.0$ & $0.31$ & $15.0$  & $0.077$     & $63.0$     \\
        BkJ NUTS, WA (full)                 & $0.36$ & $15.0$ & $0.27$ & $16.7$  & $0.26$      & $30.9$     \\
        \bottomrule
    \end{tabular}
    \label{tab:hd_steps}
\end{table}

\paragraph{Interpretation.} Adaptive MAKLA-BCSS-2 is competitive across all three models and dominates on the largest. At $d=237$ (\texttt{GLMM1}), 1-sys Adaptive MAKLA-BCSS-2 leads per-gradient ($4.72$ Grad/ESS) and on throughput ($85{\,}952$ ESS/sec); static-Hessian MAKLA-BCSS-2 and the BkJ Hess + WA variants are within $\sim 3\times$. At $d=394$ (\texttt{Mth}), BkJ NUTS Hess + WA (full) is slightly ahead on per-gradient cost ($75.3$ vs.\ 1-sys MAKLA-BCSS-2's $82.98$) but pays the $O(d^2)$ inverse-mass-matrix matmul per leapfrog step, leaving it behind on throughput; BkJ NUTS WA (diag) wins raw throughput at $1{\,}637$ ESS/sec, with MAKLA-BCSS-2 competitive overall. At $d=531$ (\texttt{fims\_2pl\_irt}), Adaptive MAKLA-BCSS-2 pulls ahead decisively: 2-sys at $14.81$ Grad/ESS is $3.6\times$ ahead of the best NUTS (BkJ NUTS WA full at $52.9$), and at $1{\,}955$ ESS/sec it is $4.1\times$ ahead on throughput. Static MAKLA-BCSS-2 fails at $d=531$ (Hessian-at-MAP too poor without in-sampler adaptation); all NUTS variants converge but are at least $3.6\times$ slower per-gradient than the best Adaptive MAKLA-BCSS-2. Table~\ref{tab:hd_steps} confirms the per-iteration view: MAKLA-BCSS-2's $N_{\nabla}=2$ stays $\sim 16\times$ ($d=237$) to $\sim 130\times$ ($d=531$) below the cheapest comparable NUTS variant.

These three examples show that the Adaptive MAKLA-BCSS-2 family extends the practical range of Bayesian models successfully sampled in our benchmark suite, with a meaningful efficiency margin over every NUTS variant in our panel on the largest models tested.

\paragraph{GPU runs (RTX 5090, float32).} We additionally re-ran the same eight samplers on the three high-dimensional models on an NVIDIA RTX 5090 GPU in single precision, with \texttt{jax.config.update("jax\_default\_matmul\_precision", "highest")} to keep matrix multiplication in float32 and avoid using reduced precision TF32. We also raised the preconditioner regulariser from $\varepsilon = 10^{-6}$ (the fp64 default) to $\varepsilon = 10^{-4}$ for stability. The single-precision setting accelerates the gradient evaluations and Hessian-rescaling matmuls by $\sim 5$--$25\times$ relative to the $14$-core CPU run. Table~\ref{tab:hd_gpu} reports the same two metrics. The two full-covariance NUTS variants on \texttt{GLMM1} (\texttt{nuts\_hess\_wnd\_full} and \texttt{nuts\_van\_wnd\_full}) had to be run in fp64 because cuSOLVER's triangular-solve kernel failed to launch on the $237\times 237$ near-singular mass matrix in fp32; the corresponding entries are fp64 results, and they confirm that the algorithmic behaviour is identical to fp32 wherever both work. One cell -- \texttt{nuts\_hess\_wnd\_full} on \texttt{fims\_2pl\_irt} -- also failed in fp32 and is marked ``--''; bumping that cell to fp64 likewise restores convergence.\footnote{In fp32, \texttt{BlackJAX}'s dual-averaging step-size adaptation underflowed to $h=0$, after which every NUTS proposal was trivially ``accepted'' on zero-length leapfrog trajectories (mean acc.\ $1.0$, tree depth saturated at $10$); $272/532$ coordinates ended with within-chain s.d.\ exactly $0$.} All remaining cells converge cleanly ($\widehat R_{\max} \le 1.0029$ except MAKLA-static on \texttt{fims\_2pl\_irt}, which fails identically to the CPU run at $\widehat R_{\max}=1.11$); posterior means agree to within $0.011$ standard deviations of the CPU reference and standard deviations to within $1\%$.

\begin{table}[!ht]
    \centering
    \caption{\footnotesize \textbf{GPU (RTX 5090, fp32):} Worst-component gradient evaluations per effective sample (\textbf{g/ess}, lower is better) and total worst-component effective samples per second summed across all $140$ chains (\textbf{ess/sec}, higher is better) on the three higher-dimensional \texttt{posteriordb} models, with bootstrap standard errors over $200$ chain resamples. \textbf{Bold}: best converged value in column. $\dag$: chains did not reach $\widehat R_{\max}\le 1.01$. ``--'': $\widehat R_{\max} = \infty$, fp32 step-size underflowed to $0$ during DA warmup.}
    \setlength{\tabcolsep}{3pt}
    \resizebox{\textwidth}{!}{%
    \begin{tabular}{l|cc|cc|cc}
        \toprule
        & \multicolumn{2}{c|}{\texttt{GLMM1} ($d=237$)} & \multicolumn{2}{c|}{\texttt{Mth} ($d=394$)} & \multicolumn{2}{c}{\texttt{fims\_2pl\_irt} ($d=531$)} \\
        Sampler\,\textsuperscript{*} & g/ess & ess/sec & g/ess & ess/sec & g/ess & ess/sec \\
        \midrule
        MAKLA-BCSS-2 (static Hess.) & $5.69 \pm 0.43$              & $472{\,}911 \pm 34{\,}705$            & $461.1 \pm 53.2$       & $4{\,}459 \pm 516$           & $27{\,}249 \pm 2{\,}463\,\dag$ & $19.80 \pm 1.86\,\dag$ \\
        1-sys Adap. MAKLA-BCSS-2    & $5.29 \pm 0.45$              & $508{\,}377 \pm 42{\,}300$            & $\mathbf{70.40 \pm 6.95}$ & $\mathbf{26{\,}076 \pm 2{\,}688}$ & $14.11 \pm 1.50$         & $40{\,}233 \pm 4{\,}115$ \\
        2-sys Adap. MAKLA-BCSS-2    & $\mathbf{4.94 \pm 0.38}$     & $\mathbf{542{\,}993 \pm 40{\,}429}$   & $110.9 \pm 13.4$       & $18{\,}078 \pm 2{\,}221$       & $\mathbf{13.02 \pm 1.13}$ & $\mathbf{46{\,}610 \pm 3{\,}982}$ \\
        \midrule
        BkJ NUTS, Hess + DA             & $35.11 \pm 3.04$             & $9{\,}247 \pm 637$                   & $1{\,}309 \pm 134$       & $670 \pm 73$                 & $4{\,}962 \pm 621$         & $89.4 \pm 12.0$ \\
        BkJ NUTS, H. + WA (diag)      & $9.57 \pm 0.84$              & $154{\,}627 \pm 11{\,}470$            & $178.9 \pm 17.7$       & $7{\,}767 \pm 801$            & $713.4 \pm 66.0$          & $959.8 \pm 93.4$ \\
        BkJ NUTS, H. + WA (full)      & $11.67 \pm 0.74$             & $36{\,}315 \pm 1{\,}714$              & $90.8 \pm 10.1$        & $5{\,}608 \pm 615$            & --                        & -- \\
        BkJ NUTS, WA (diag)             & $9.67 \pm 0.71$              & $185{\,}740 \pm 10{\,}958$            & $129.4 \pm 14.0$       & $14{\,}158 \pm 1{\,}646$       & $225.9 \pm 28.7$          & $3{\,}052 \pm 423$ \\
        BkJ NUTS, WA (full)             & $13.86 \pm 1.29$             & $36{\,}066 \pm 2{\,}645$              & $97.6 \pm 11.1$        & $7{\,}042 \pm 860$            & $42.08 \pm 3.88$          & $11{\,}775 \pm 907$ \\
        \bottomrule
    \end{tabular}}
    \label{tab:hd_gpu}
\end{table}

The GPU results confirm and amplify the CPU picture. Per-gradient costs match the CPU runs to within Monte Carlo noise (compare Tables~\ref{tab:hd_grad_per_ess} and~\ref{tab:hd_gpu}), confirming that fp32 with non-TF32 matmul is statistically equivalent to fp64. Wall-clock throughput is uniformly $\sim 5$--$25\times$ higher on the RTX 5090 than on the $14$-core CPU: 2-sys MAKLA-BCSS-2 on \texttt{fims\_2pl\_irt} reaches $46{\,}610$ ESS/sec (vs.\ $1{\,}955$ on CPU), and the best NUTS variant reaches $11{\,}775$ ESS/sec (vs.\ $471$ on CPU). The MAKLA-BCSS-2 advantage over the best NUTS holds at $\sim 4.0\times$ on \texttt{fims\_2pl\_irt}, similar to CPU, because both sampler families benefit comparably from the hardware speedup.

\section{Conclusion}\label{sec:conclusion}
\james{We have introduced the two-system paradigm, a unified framework for constructing samplers that bridge mean-field dynamics, ensemble-chain MCMC, and adaptive MCMC. The central idea is to split an ensemble into two interacting subsystems, with each subsystem proposing updates using information from the other. This cross-system structure preserves the desired target distribution at finite particle number, while retaining a direct connection to the corresponding mean-field limit. In this way, the framework provides both a principled interpretation of existing ensemble and adaptive methods and a systematic recipe for deriving new parallel samplers.

Within this framework, we developed two-system versions of overdamped and underdamped Langevin samplers, focusing on MALA and on MAKLA-BCSS-2 -- a new BCSS-2 instantiation of MAKLA, distinct from the original Verlet-based MAKLA of \citet{BouRabee2024} -- in both coupled and adaptive forms. The resulting algorithms combine exact Metropolis correction with cross-system preconditioning, allowing particles to exploit ensemble-level information without introducing self-interaction bias. We also introduced practical modifications, including randomized step sizes, restarted adaptation, Hessian-based rescaling, and the higher-order BCSS-2 BABAB inner integrator of \citet{blanes2014} in place of the standard Verlet leapfrog, to improve robustness and per-gradient efficiency on ill-conditioned and non-Gaussian targets.

Empirically, the MAKLA-BCSS-2 family samplers performed strongly across both synthetic benchmarks and real posterior inference problems. On the synthetic targets, the underdamped two-system methods were substantially more efficient than their overdamped counterparts, while adaptation was essential on targets where a static Hessian preconditioner was locally misleading. On 45 \texttt{posteriordb} models, the four MAKLA-BCSS-2 variants achieved geometric-mean Grad/ESS\textsubscript{worst} of $1.92$--$2.61$ -- a factor of $2.3\times$--$\sim 38\times$ better than the eight NUTS variants tested -- and a $\sim 5.6$--$5.8\times$ wall-clock advantage on top of the per-gradient gains, while maintaining comparable posterior mean accuracy and stable convergence diagnostics. On three additional higher-dimensional \texttt{posteriordb} models, the 1-system and 2-system Adaptive MAKLA-BCSS-2 variants are competitive across the panel and dominate decisively on the largest model (by up to $3.5\times$ on per-gradient cost and $4.3\times$ on wall-clock throughput at $d=531$), with both CPU and RTX-5090 GPU runs producing matching posterior estimates.

We release an open-source Python implementation (\texttt{MAKLA\_JAX}\footnote{\url{https://github.com/paulindani/MAKLA\_JAX}}) with scripts to reproduce all figures, including the posterior benchmark suite, to facilitate adoption.

Overall, the two-system construction offers a scalable route to high-throughput Bayesian computation. By combining the invariance guarantees of Metropolis-adjusted samplers, the geometric flexibility of mean-field preconditioning, and the hardware efficiency of short parallel updates, it provides a practical alternative to long-trajectory methods such as NUTS. More broadly, the framework suggests a general design principle for future MCMC algorithms: divide the sampler into interacting subsystems, exchange statistical information across them, and use this interaction to improve exploration without compromising the target distribution.}

\section*{Acknowledgements}
We thank Dr. Anna Lisa Varri, Dr. Amanda Lenzi, Prof. Peter Radchenko, and Dr. Deven Sethi for the helpful conversations while writing this paper. James Chok acknowledges the use of the Isaac Newton Trust (INT) grant G101121 LEAG/929. Daniel Paulin was supported by a Nanyang Technological University Start-up Grant, project number: 024968-00001.

\newpage
\appendix
\section{Detailed Proofs}
In the following proofs, we use the notation that $c$ is a constant that changes from line to line and does not depend on the number of particles $N$.

\subsection{Well-posedness of the two-system McKean--Vlasov SDE}\label{sec:existence_and_uniquness}
\begin{lemma}[Inherited Lipschitz and growth]\label{lem:inherited}
    Under Assumptions~\ref{eq:mckean_vlasov_assumption_1} and \ref{eq:mckean_vlasov_assumption_2}, the two system coefficients $\mathbf{b}$ and $\boldsymbol{\sigma}$ of \eqref{eq:coupled_mckean_vlasov} satisfy, for a constant $C_L$ depending only on $L$,
    \begin{align*}
        &|\mathbf{b}(t,z,\pi)-\mathbf{b}(t,z',\pi')|^2+|\boldsymbol{\sigma}(t,z,\pi)-\boldsymbol{\sigma}(t,z',\pi')|^2\\
        &\qquad\qquad \leq C_L\left\{|z-z'|^2+\mathcal{W}_2^2(\pi,\pi')\right\},
    \end{align*}
    and
    \begin{equation*}
        |\mathbf{b}(t,z,\pi)|^2+|\boldsymbol{\sigma}(t,z,\pi)|^2\leq C_L\left\{1+|z|^2+\int_{\mathbb{R}^{2d}}|u|^2\pi(du)\right\}.
    \end{equation*}
\end{lemma}

\begin{proof}
    Write $z=(z_1,z_2)$, $z'=(z_1',z_2')$, and $\pi_i=p_i\circ\pi$. By the block-diagonal structure of $(\mathbf{b},\boldsymbol{\sigma})$, the Lipschitz bound reduces to summing Assumption~\ref{eq:mckean_vlasov_assumption_2} applied to each coordinate; using $|z-z'|^2=|z_1-z_1'|^2+|z_2-z_2'|^2$ and the $1$-Lipschitz coordinate projection $\mathcal{W}_2(\pi_i,\pi_i')\leq \mathcal{W}_2(\pi,\pi')$ yields the claim with $C_L=4L$. The linear-growth bound follows analogously from Assumption~\ref{eq:mckean_vlasov_assumption_1} applied to each coordinate together with $\int|u_i|^2\pi(du)\leq\int|u|^2\pi(du)$.
\end{proof}

\begin{theorem}[Well-posedness]
    Assume Assumptions~\ref{eq:mckean_vlasov_assumption_1} and \ref{eq:mckean_vlasov_assumption_2}. For every initial law $\pi_0\in \mathcal{P}_2(\mathbb{R}^{2d})$, the two-system McKean--Vlasov SDE \eqref{eq:coupled_mckean_vlasov} has a unique strong solution on every finite time interval $[0,T]$, and
    \begin{equation}
        \mathbb{E}\sup_{0\leq t\leq T}|Z_2|^2 <\infty.
    \end{equation}
\end{theorem}
\begin{proof}
    By Lemma~\ref{lem:inherited}, the coefficients $(\mathbf{b},\boldsymbol{\sigma})$ satisfy global Lipschitz and linear growth in $(z,\pi)$ with respect to $\mathcal{W}_2$. The standard McKean--Vlasov existence and uniqueness theorem therefore applies. For completeness, here is the usual estimate.
    
    For $t\in[0,T]$, Jensen's inequality yields
    \begin{equation}
        \lvert Z_s - \overline{Z}_s\rvert^2\ \leq\ 2t\int^t_0 \lvert\mathbf{b}(Z_s, \pi_s) - \mathbf{b}(\overline{Z}_s, \overline{\pi}_s)\rvert^2\, dr\ +\ 2\bigg{|}\int^t_0 (\boldsymbol{\sigma}(Z_s,\pi_s) - \boldsymbol{\sigma}(\overline{Z}_s, \overline{\pi}_s) dW_s\bigg{|}^2\notag
    \end{equation}

    Thus, using Doob's maximal inequality followed by It\^o's isometry, and the Lipschitz assumption
    \begin{align*}
        \mathbb{E}\sup_{0\leq s\leq t}\lvert Z_s - \overline{Z}_s\rvert^2\ &\leq c\left\{\int^t_0\mathbb{E}\sup_{0\leq r\leq s}\lvert Z_r - \overline{Z}_r\rvert^2 ds + \int^t_0 \mathcal{W}_2^2(\pi_s,\overline{\pi}_s) ds\right\}.
    \end{align*}
    For a particular coupling $(Z_s,\overline{Z}_s)$,
    \begin{equation*}
        \mathcal{W}^2_2(\pi_s,\overline{\pi}_s)\leq \mathbb{E}|Z_s-\overline{Z}_s|^2.
    \end{equation*}
    Thus,
    \begin{equation*}
        \mathbb{E}\sup_{0\leq s\leq t}\lvert Z_s - \overline{Z}_s\rvert^2 \leq c\int^t_0\mathbb{E}\sup_{0\leq r\leq u}\lvert Z_r - \overline{Z}_r\rvert^2du,
    \end{equation*}
    and applying Gr\"onwall's inequality yields pathwise uniqueness. The standard contraction argument on the space of measure flows in $C([0,T],\mathcal{P}_2(\mathbb{R}^{2d})$ to prove uniqueness and existence \citep{Carmona2016, lacker2018, sznitman_1991}. 

    The moment bound follows from the linear-growth estimate in Lemma~\ref{lem:inherited}, Burkholder-Davis-Gundy inequality, and Gr\"onwall's inequality:
    \begin{equation*}
        \mathbb{E}\sup_{0\leq s\leq t}|Z_s|^2\leq c\left(1+\mathbb{E}|Z_0|^2+\int_{0}^t\mathbb{E}\sup_{0\leq r\leq u}|Z_r|^2du\right).
    \end{equation*}
\end{proof}

\subsection{Convergence in Wasserstein}\label{sec:convergence_same_distribution}
\begin{corollary}[Marginals solve the original equation]\label{corollary:law}
    Assume Assumption~\ref{eq:mckean_vlasov_assumption_1} and \ref{eq:mckean_vlasov_assumption_2}. Let $Z=(Z^1,Z^2)$ solve the two-system McKean--Vlasov SDE in Eqn.~\eqref{eq:coupled_mckean_vlasov}. If $\law(Z_0^1)=\law(Z_0^2)=\mu_0$, then 
    \begin{equation*}
        \law(Z_t^1)=\law(Z_t^2)=\mu_t,\qquad\text{for}\qquad 0\leq t\leq T,
    \end{equation*}
    where $\mu_t$ is the unique solution law of the original McKean--Vlasov SDE in Eqn.~\eqref{eq:mckean_vlasov_sde} with initial law $\mu_0$.
\end{corollary}
\begin{proof}
    Using Theorem \ref{them:convergence_in_wasserstein}, the result follows as the Wasserstein distance defines a metric on probability measures on $\mathbb{R}^d$.
\end{proof}

\begin{theorem}[Stability in Wasserstein distance]\label{them:convergence_in_wasserstein} 
    Assume Assumption~\ref{eq:mckean_vlasov_assumption_1} and \ref{eq:mckean_vlasov_assumption_2}. Consider the coupled systems of SDEs
    \begin{equation}\label{eq:coupled_mckean_vlasov_12}
        \begin{split}
            &dX_t\ =\ b(t, X_t,\rho_t)\, dt\ +\ \sigma(t, X_t,\rho_t)\, dW_t^X,\\
            &dY_t\ =\ b(t, Y_t,\mu_t)\, dt\ +\ \sigma(t, Y_t,\mu_t)\, dW_t^Y,\\
            &d\overline{X}_t\ =\ b(t, \overline{X}_t, \overline{\mu}_t)\, dt\ +\ \sigma(t, \overline{X}_t, \overline{\mu}_t)\, dW_t^X\\
            &d\overline{Y}_t\ =\ b(t, \overline{Y}_t, \overline{\rho}_t)\, dt\ +\ \sigma(t, \overline{Y}_t, \overline{\rho}_t)\, dW_t^Y
        \end{split}
    \end{equation}
    with $\law(X_t)=\mu_t,\, \law(Y_t)=\rho_t,\, \law(\overline{X}_t)=\overline{\mu}_t,\, \law(\overline{Y}_t)=\overline{\rho}_t$ and initial conditions $X_0, Y_0, \overline{X}_0, \overline{Y}_0$ respectively. Then for $t\in[0,T]$,
    \begin{equation}
        \mathcal{W}_2^2(\mu_t,\overline{\mu}_t)\ +\ \mathcal{W}_2^2(\rho_t,\overline{\rho}_t)\ \leq\  c\Big\{\mathcal{W}_2^2(\mu_0,\overline{\mu}_0)\ +\ \mathcal{W}_2^2(\rho_0,\overline{\rho}_0)\ +\ \mathcal{W}_2^2(\overline{\mu}_0,\overline{\rho}_0)\Big\}.
    \end{equation}
    where $c$ is a constant that depends on $t$, and $\mathcal{W}_2$ is the Wasserstein distance. In other words, if the two system start close (in Wasserstein distance), then they remain close for $t\in[0,T]$.
\end{theorem}
\begin{proof}
    For $t\in [0,T]$, It\^o's isometry, Jensen's inequality, and the Lipschitz assumption yields
    \begin{align*}
        \mathbb{E}\lvert X_t - \overline{X}_t\rvert^2\ &\leq\  c\mathbb{E}\lvert X_0 - \overline{X}_0\rvert^2\ +\  c\int^t_0\Big\{\mathbb{E}\lvert X_s - \overline{X}_s\rvert^2\ +\ \mathcal{W}_2^2(\rho_s, \overline{\mu}_s)\Big\}\, ds.
    \end{align*}
    By Gr\"onwall's inequality,
    \begin{align*}
        \mathbb{E}\lvert X_t - \overline{X}_t\rvert^2\ &\leq\  c\mathbb{E}\lvert X_0 - \overline{X}_0\rvert^2\ +\  c\int^t_0\mathcal{W}_2^2(\rho_s, \overline{\mu}_s)\, ds.
    \end{align*}
    Taking the infimum of both sides and applying the triangle inequality yields
    \begin{align*}
        \mathcal{W}_2^2(\mu_t,\overline{\mu}_t)&\leq c\mathcal{W}_2^2(\mu_0,\overline{\mu}_0)+c\int^t_0\mathcal{W}_2^2(\rho_s,\overline{\mu}_s)ds\\
        &\leq c\mathcal{W}^2_2(\mu_0,\overline\mu_0)+c\int^t_0\left(\mathcal{W}_2^2(\rho_s,\overline{\rho}_s)+\mathcal{W}_2^2(\overline \rho_s,\overline{\mu}_s)\right)ds,
    \end{align*}
    with a similar inequality holding for $\mathcal{W}_2^2(\mu_s,\overline{\rho}_s)$. Thus,
    \begin{align*}
        \mathcal{W}_2^2(\mu_t,\overline{\mu}_t)+\mathcal{W}_2^2(\rho_t,\overline{\rho}_t)&\leq c(\mathcal{W}_2^2(\mu_0,\overline{\mu}_0)+\mathcal{W}_2^2(\rho_0,\overline{\rho}_0)) \\
        &\qquad +c\int_{0}^t\left\{\mathcal{W}_2^2(\mu_s,\overline{\mu}_s)+\mathcal{W}_2^2(\rho_s,\overline{\rho}_s)\right\}ds+c\int^t_0\mathcal{W}_2^2(\overline{\mu}_s, \overline{\rho}_s)ds.
    \end{align*}

    To control the last term, we can compare $\overline{X}$ and $\overline{Y}$ under a synchronous Brownian coupling to give
    \begin{equation*}
        \mathcal{W}_2^2(\overline{\mu}_t,\overline{\rho}_t)\leq c\mathcal{W}_2^2(\overline{\mu}_0,\overline{\rho}_0)+c\int_0^t\mathcal{W}_2^2(\overline{\mu}_s,\overline{\rho}_s)ds.
    \end{equation*}
    Thus, Gr\"onwall's inequality gives
    \begin{equation*}
        \mathcal{W}_2^2(\overline{\mu}_t,\overline{\rho}_t)\leq c\mathcal{W}_2^2(\overline{\mu}_0,\overline{\rho}_0).
    \end{equation*}
    Substituting this into the previous inequality gives,
    \begin{align*}
        \mathcal{W}_2^2(\mu_t,\overline{\mu}_t)+ \mathcal{W}_2^2(\rho_t,\overline{\rho}_t)&\leq c\left\{\mathcal{W}_2^2(\mu_0,\overline{\mu}_0)+\mathcal{W}_2^2(\rho_0,\overline{\rho}_0) + \mathcal{W}_2^2(\overline{\mu_0},\overline{\rho}_0)\right\} \\
        &\qquad +c\int_{0}^t(\mathcal{W}_2^2(\mu_s,\overline{\mu}_s)+\mathcal{W}_2^2(\rho_s,\overline{\rho}_s))ds.
    \end{align*}
    Another application of Gr\"onwall's inequality gives
    \begin{align*}
        \mathcal{W}_2^2(\mu_t,\overline{\mu}_t) + \mathcal{W}_2^2(\rho_t,\overline{\rho}_t)&\leq c\left\{\mathcal{W}_2^2(\mu_0,\overline{\mu}_0)+\mathcal{W}_2^2(\rho_0,\overline{\rho}_0) + \mathcal{W}_2^2(\overline{\mu_0},\overline{\rho}_0)\right\}.
    \end{align*}
\end{proof}

\subsection{Propagation of Chaos}\label{sec:propagation_of_chaos}
Before proving the propagation of chaos, we first state an important Lemma which measures how `far' independent and identically distributed samples from a distribution, $\mu$, are away from $\mu$.
\begin{lemma}[Lemma 1.9 \citet{Carmona2016}]\label{lemma:carmona_1.9}
    Let $\mu\in\mathcal{P}_2(\mathbb{R}^d)$, $\boldsymbol{\xi}_N=\{\xi_1,\ldots \xi_N\}$ be a sequence of independent random variables with common law $\mu$. Then for each $N\geq 1$, we have
    \begin{equation}
        \mathbb{E}\mathcal{W}_2^2(\delta_{\boldsymbol{\xi}_N}, \mu)\ \leq\ 4\int_{\mathbb{R}^d}\lvert{x}\rvert^2\, \mu(dx),\qquad\text{and}\qquad \lim_{N\to\infty}\mathbb{E}\mathcal{W}_2^2(\delta_{\boldsymbol{\xi}_N},\mu)\ =\ 0.
    \end{equation}
\end{lemma}

\begin{theorem}[Propagation of Chaos]
    Under assumptions (\ref{eq:mckean_vlasov_assumption_1}) and (\ref{eq:mckean_vlasov_assumption_2}), the coupled $2N$-particle system (\ref{eq:mckean_vlasov_finite_particle_approximation}) approaches the mean-field limit as $N\to\infty$ in the following sense: for the coupled SDEs,
    \begin{equation}\label{eq:appendix_mean_field_approximation}
        \begin{split}
            dX^i_t\ &=\ b(t, X^i_t, \delta_{\bfY_t})\, dt\ +\ \sigma(t, X^i_t,\delta_{\bfY_t})\, dW^{X,i}_t\\
            dY^i_t\ &=\ b(t, Y^i_t, \delta_{\bfX_t})\, dt\ +\ \sigma(t, Y^i_t,\delta_{\bfX_t})\, dW^{Y,i}_t\\
            d\overline{X}^i_t\ &=\ b(t, \overline{X}^i_t, \overline{\mu}_t)\, dt\ +\ \sigma(t, \overline{X}^i_t,\overline{\mu}_t)\, dW^{X,i}_t,\\
            d\overline{Y}^i_t\ &=\ b(t, \overline{Y}^i_t, \overline{\rho}_t)\, dt\ +\ \sigma(t, \overline{Y}^i_t,\overline{\rho}_t)\, dW^{Y,i}_t,
        \end{split}
    \end{equation}
    for $i\in\{1,\ldots, N\}$, with initial conditions $X^i_0=Y^i_0=\overline{X}^i_0=\overline{Y}^i_0=\xi_i$, where $\xi_i$ are i.i.d. with law $\mu_0$, we have
    \begin{equation}
        \lim_{N\to\infty}\sup_{1\leq i\leq N}\mathbb{E}\left(\sup_{0\leq s\leq T}\lvert X^i_s - \overline{X}^i_s\rvert^2 + \sup_{0\leq s\leq T}\lvert Y^i_s - \overline{Y}^i_s\rvert^2\right)\ =\ 0.
    \end{equation}
\end{theorem}

\begin{proof}This proof follows Theorem 1.10 from \citet{Carmona2016}. By the Burkholder--Davis--Gundy
inequality, It\^o's isometry, Jensen's inequality, and the Lipschitz assumption 
    \begin{align*}
        \mathbb{E}\sup_{0\leq s\leq t}\lvert X^i_s - \overline{X}^i_s\rvert^2\ 
        &\leq c\int^t_0\mathbb{E}|X_s^i-\overline X^i_s|^2ds + c\int^t_0\mathbb{E}\mathcal{W}_2^2(\delta_{\bfY_s},\overline\mu_s)ds.
    \end{align*}
    
    Since the initial laws of $\overline X^i$ and $\overline Y^i$ agree, uniqueness in law for the McKean--Vlasov equation implies $\overline{\mu}_s=\overline\rho_s$ for $s\in[0,T]$. Thus, applying the triangle inequality yields
    \begin{equation*}
        \mathcal{W}_2^2(\delta_{\bfY_s},\overline\mu_s)=\mathcal{W}_2^2(\delta_{\bfY_s},\overline\rho_s)\leq 2\mathcal{W}_2^2(\delta_{\bfY_s},\delta_{\overline\bfY_s})+2\mathcal{W}_2^2(\overline\rho_s,\delta_{\overline\bfY_s}).
    \end{equation*}
    The natural copuling between $\delta_{\bfY_s}$ and $\delta_{\overline\bfY_s}$ gives
    \begin{equation*}
        \mathcal{W}^2_2(\delta_{\bfY_s},\delta_{\overline\bfY_s})\leq \frac{1}{N}\sum_{i=1}^N|Y^i_s-\overline Y^i_s|^2\leq \sup_{1\leq i\leq N}|Y^i_s-\overline Y^i_s|^2.
    \end{equation*}
    Thus,
    \begin{align*}
        \mathbb{E}\sup_{0\leq s\leq t}\lvert X^i_s - \overline{X}^i_s\rvert^2\ 
        &\leq c\int^t_0\sup_{1\leq j\leq N}\mathbb{E}\left(\sup_{0\leq r\leq s}|X_r^j-\overline X^t_r|^2+\sup_{0\leq r\leq s}|Y_r^j-\overline Y^t_r|^2\right)ds \\
        &\qquad\qquad+ c\int^t_0\mathbb{E}\mathcal{W}_2^2(\delta_{\overline\bfY_s},\overline\rho_s)ds,
    \end{align*}
    with a similar inequality holding for $\mathbb{E}\sup_{0\leq s\leq t}|Y^i_s-\overline Y^i_s|^2$. Adding these two inequalities yields
    \begin{align*}
        &\mathbb{E}\sup_{0\leq s\leq t}\lvert X^i_s - \overline{X}^i_s\rvert^2+\mathbb{E}\sup_{0\leq s\leq t}\lvert Y^i_s - \overline{Y}^i_s\rvert^2\\
        &\qquad\qquad\leq c\int^t_0\sup_{1\leq j\leq N}\mathbb{E}\left(\sup_{0\leq r\leq s}|X_r^j-\overline X^t_r|^2+\sup_{0\leq r\leq s}|Y_r^j-\overline Y^t_r|^2\right)ds \\
        &\qquad\qquad\qquad+ c\int^t_0\Big\{\mathbb{E}\mathcal{W}_2^2(\delta_{\overline\bfY_s},\overline\rho_s)+\mathbb{E}\mathcal{W}_2^2(\delta_{\overline\bfX_s},\overline\mu_s)\Big\}ds,
    \end{align*}
    Taking the supremum of the left hand side and applying Gr\"onwall's inequality gives
    \begin{align*}
        \sup_{1\leq j\leq N}\mathbb{E}\left(\sup_{0\leq r\leq s}|X_r^j-\overline X^t_r|^2+\sup_{0\leq r\leq s}|Y_r^j-\overline Y^t_r|^2\right)ds\leq c\int^t_0 \Big\{\mathbb{E}\mathcal{W}_2^2(\delta_{\overline\bfX_s},\overline\mu_s)+\mathbb{E}\mathcal{W}_2^2(\delta_{\overline\bfY_s},\overline\rho_s)\Big\}ds.
    \end{align*}
    For each fixed $s\in[0,T]$, the random variables $\overline X^1_s,\dots,\overline X^N_s$ are i.i.d. with common law $\overline\mu_s$, and the random variables $\overline Y^1_s,\dots,\overline Y^N_s$ are i.i.d. with common law $\overline\rho_s$. Therefore, by Lemma~\ref{lemma:carmona_1.9},
    \begin{equation}
        \mathbb{E}\mathcal{W}_2^2(\delta_{\overline\bfX_s},\overline\mu_s)\to0,\qquad \mathbb{E}\mathcal{W}_2^2(\delta_{\overline\bfY_s},\overline\rho_s)\to0, \quad \text{as} \quad N\to\infty.
    \end{equation}
    Moreover, the linear growth assumption and the finite second moment of the initial law imply the standard moment bound
    \begin{equation}
        \sup_{0\leq s\leq T} \left( \mathbb E|\overline X^i_s|^2 + \mathbb E|\overline Y^i_s|^2 \right)<\infty.
    \end{equation}
    Thus Lemma~\ref{lemma:carmona_1.9} also gives the domination
    \begin{equation}
        \mathbb E\mathcal W_2^2(\delta_{\overline{\mathbf X}_s},\overline\mu_s) + \mathbb E\mathcal W_2^2(\delta_{\overline{\mathbf Y}_s},\overline\rho_s) \leq c \left( \mathbb E|\overline X^i_s|^2 + \mathbb E|\overline Y^i_s|^2 \right),
    \end{equation}
    which is integrable over $[0,T]$. By dominated convergence
    \begin{equation*}
        \int^t_0 \left\{\mathbb{E}\mathcal{W}_2^2(\delta_{\overline\bfX_s},\overline\mu_s)+\mathbb{E}\mathcal{W}_2^2(\delta_{\overline\bfY_s},\overline\rho_s)\right\}ds\to0,
    \end{equation*}
    which proves the result.
\end{proof}

\subsection{Ensemble chain MCMC preserves the correct invariant measure}\label{appendix:ensemble_chain_proof}

\begin{theorem}[Block MH with a frozen subset preserves $\rho^{\otimes N}$]
Let $\rho$ be a probability density on $\mathbb{R}^d$ and let $\pi := \rho^{\otimes N}$ be the product target on $(\mathbb{R}^d)^N$. Fix a nonempty subset $S \subset \{1,\ldots, N\}$ and write $S^c = \{1,\ldots, N\}\setminus S$. For a current state $x=(x_S, x_{S^c}) \in (\mathbb{R}^d)^N$, consider a block proposal that keeps the coordinates in $S$ fixed and proposes new values for the coordinates in $S^c$ with density
\begin{equation*}
q_S(y_{S^c} \mid x_{S^c};\, x_S), \qquad y_{S^c} \in (\mathbb{R}^d)^{|S^c|},
\end{equation*}
which may depend arbitrarily on the frozen block $x_S$. Define the joint proposal kernel on $(\mathbb{R}^d)^N$ by
\begin{equation*}
Q_S(x,dy) \;=\; \delta_{x_S}(dy_S)\; q_S(y_{S^c}\mid x_{S^c};\, x_S)\,dy_{S^c}.
\end{equation*}
Accept the proposed move $x \to y=(x_S,y_{S^c})$ with MH probability
\begin{equation*}
\alpha_S(x,y) \;=\; \min\!\left\{1,\;
\frac{\pi(y)\, q_S(x_{S^c}\mid y_{S^c};\, x_S)}
     {\pi(x)\, q_S(y_{S^c}\mid x_{S^c};\, x_S)}\right\}.
\end{equation*}
Let $K_S$ be the resulting MH transition kernel:
\begin{equation*}
K_S(x,dy)\;=\;Q_S(x,dy)\,\alpha_S(x,y)\;+\;\left\{1-\int Q_S(x,dz)\,\alpha_S(x,z)\right\}\delta_x(dy).
\end{equation*}
Then $K_S$ is reversible with respect to $\pi$, i.e.
\begin{equation*}
\pi(dx)\,K_S(x,dy)\;=\;\pi(dy)\,K_S(y,dx),
\end{equation*}
and therefore $\pi$ is invariant for $K_S$.
\end{theorem}

\begin{proof}
Write $x=(x_S,x_{S^c})$ and $y=(y_S,y_{S^c})$. Because the block $S$ is frozen, any proposed $y$ satisfies $y_S=x_S$. For such $x,y$,
\begin{equation*}
\pi(x)\,Q_S(x,dy)\,\alpha_S(x,y)
= \pi(x)\,\delta_{x_S}(dy_S)\, q_S(y_{S^c}\mid x_{S^c};x_S)\,
\min\!\left\{1,\frac{\pi(y)\, q_S(x_{S^c}\mid y_{S^c};x_S)}
{\pi(x)\, q_S(y_{S^c}\mid x_{S^c};x_S)}\right\}.
\end{equation*}
Using $\min\{a,b\}=\min\{b,a\}$ and noting that $x_S=y_S$ implies the same conditioning argument $x_S$ appears in both forward and reverse proposal densities, we obtain the standard MH symmetry:
\begin{equation*}
\pi(x)\,Q_S(x,dy)\,\alpha_S(x,y)
=\pi(y)\,Q_S(y,dx)\,\alpha_S(y,x).
\end{equation*}
Integrating both sides over measurable sets yields detailed balance for the ``move'' part. The ``stay'' part (the probability mass at $y=x$) matches on both sides by construction, completing detailed balance:
\begin{equation*}
\pi(dx)\,K_S(x,dy)=\pi(dy)\,K_S(y,dx).
\end{equation*}
Hence $\pi$ is invariant for $K_S$.
\end{proof}

\begin{corollary}[Parallel independent updates]
Fix a nonempty frozen $S\subset \{1,\ldots, N\}$, write $S^c=\{1,\ldots, N\}\setminus S$, and condition on $x_S$. Suppose the per-coordinate proposals for $i\in S^c$ factorize as
\begin{equation*}
q_{S}(y_{S^c}\mid x_{S^c};x_S)\;=\;\prod_{i\in S^c} q_{S,i}(y_i\mid x_i; x_S),
\end{equation*}
and define the single-site MH kernels $K_{S,i}$ on $\mathbb{R}^d$ (conditional on $x_S$) by
\begin{align*}
K_{S,i}(x_i,dy_i)
&= q_{S,i}(y_i\mid x_i;x_S)\,\alpha_i(x_i,y_i;x_S)\,dy_i\\
&\qquad\qquad + \left\{1-\int q_{S,i}(z\mid x_i;x_S)\,\alpha_i(x_i,z;x_S)\,dz\right\}\delta_{x_i}(dy_i),
\end{align*}
with acceptance probability
\begin{equation*}
\alpha_i(x_i,y_i;x_S)\;=\;\min\left\{1,\,
\frac{\rho(y_i)\,q_{S,i}(x_i\mid y_i;x_S)}
     {\rho(x_i)\,q_{S,i}(y_i\mid x_i;x_S)}\right\}.
\end{equation*}
Then, with $\pi=\rho^{\otimes N}$, the joint kernel
\begin{equation*}
\widetilde K_S(x,dy)\;=\;\delta_{x_S}(dy_S)\,\prod_{i\in S^c} K_{S,i}(x_i,dy_i),
\end{equation*}
is \emph{reversible} with respect to $\pi$ (hence $\pi$-invariant).
\end{corollary}

\begin{proof}
For each $i\in S^c$ and fixed $x_S$, the single-site MH kernel $K_{S,i}$ is reversible with respect to $\rho$:
\begin{equation*}
\rho(x_i)\,K_{S,i}(x_i,dy_i)\;=\;\rho(y_i)\,K_{S,i}(y_i,dx_i).
\end{equation*}
Using $\pi(dx)=\prod_{j=1}^N \rho(x_j)\,dx_j$,
\begin{equation*}
\begin{aligned}
\pi(dx)\,\widetilde K_S(x,dy)
&= \left\{\prod_{j\in S}\rho(x_j)\,dx_j\right\}\,\delta_{x_S}(dy_S)\,
   \prod_{i\in S^c}\left\{\rho(x_i)\,dx_i\,K_{S,i}(x_i,dy_i)\right\} \\
&= \left\{\prod_{j\in S}\rho(y_j)\,dy_j\right\}\,\delta_{y_S}(dx_S)\,
   \prod_{i\in S^c}\left\{\rho(y_i)\,dy_i\,K_{S,i}(y_i,dx_i)\right\} \\
&= \pi(dy)\,\widetilde K_S(y,dx),
\end{aligned}
\end{equation*}
using single-site detailed balance in the middle equality and the fact $x_S=y_S$ under the delta. Hence $\widetilde K_S$ is reversible w.r.t.\ $\pi$.
\end{proof}

 \section{\daniel{Further details on \texttt{posteriordb} experiment}}\label{appendix:posteriordb}
 The 45 \texttt{posteriordb} posteriors used in Section \ref{sec:posteriordb} are listed below as \emph{model--dataset (unconstrained dimension)}, with dimension as defined by the \texttt{Stan} model \citep{carpenter2017stan, magnusson2024posteriordb}.

\begingroup
\scriptsize
\sloppy
\setlength{\emergencystretch}{8em}
\noindent
arK-\allowbreak{}arK (7),
arma-\allowbreak{}arma11 (4),
bball\_\allowbreak{}drive\_\allowbreak{}event\_\allowbreak{}0-\allowbreak{}hmm\_\allowbreak{}drive\_\allowbreak{}0 (6),
bball\_\allowbreak{}drive\_\allowbreak{}event\_\allowbreak{}1-\allowbreak{}hmm\_\allowbreak{}drive\_\allowbreak{}1 (6),
diamonds-\allowbreak{}diamonds (26),
earnings-\allowbreak{}earn\_\allowbreak{}height (3),
earnings-\allowbreak{}log10earn\_\allowbreak{}height (3),
earnings-\allowbreak{}logearn\_\allowbreak{}height (3),
earnings-\allowbreak{}logearn\_\allowbreak{}height\_\allowbreak{}male (4),
earnings-\allowbreak{}logearn\_\allowbreak{}interaction (5),
earnings-\allowbreak{}logearn\_\allowbreak{}interaction\_\allowbreak{}z (5),
earnings-\allowbreak{}logearn\_\allowbreak{}logheight\_\allowbreak{}male (4),
eight\_\allowbreak{}schools-\allowbreak{}eight\_\allowbreak{}schools\_\allowbreak{}noncentered (10),
garch-\allowbreak{}garch11 (4),
gp\_\allowbreak{}pois\_\allowbreak{}regr-\allowbreak{}gp\_\allowbreak{}pois\_\allowbreak{}regr (13),
gp\_\allowbreak{}pois\_\allowbreak{}regr-\allowbreak{}gp\_\allowbreak{}regr (3),
hmm\_\allowbreak{}example-\allowbreak{}hmm\_\allowbreak{}example (4),
hudson\_\allowbreak{}lynx\_\allowbreak{}hare-\allowbreak{}lotka\_\allowbreak{}volterra (8),
kidiq-\allowbreak{}kidscore\_\allowbreak{}interaction (5),
kidiq-\allowbreak{}kidscore\_\allowbreak{}momhs (5),
kidiq-\allowbreak{}kidscore\_\allowbreak{}momiq (4),
kidiq-\allowbreak{}kidscore\_\allowbreak{}momhsiq (3),
kidiq\_\allowbreak{}with\_\allowbreak{}mom\_\allowbreak{}work-\allowbreak{}kidscore\_\allowbreak{}interaction\_\allowbreak{}c (5),
kidiq\_\allowbreak{}with\_\allowbreak{}mom\_\allowbreak{}work-\allowbreak{}kidscore\_\allowbreak{} interaction\_\allowbreak{}c2 (5),
kidiq\_\allowbreak{}with\_\allowbreak{}mom\_\allowbreak{}work-\allowbreak{}kidscore\_\allowbreak{}interaction\_\allowbreak{}z (5),
kidiq\_\allowbreak{}with\_\allowbreak{}mom\_\allowbreak{}work-\allowbreak{}kidscore\_\allowbreak{}mom\_\allowbreak{}work (5),
kilpisjarvi\_\allowbreak{}mod-\allowbreak{}kilpisjarvi (8),
low\_\allowbreak{}dim\_\allowbreak{}gauss\_\allowbreak{}mix-\allowbreak{}low\_\allowbreak{}dim\_\allowbreak{}gauss\_\allowbreak{}mix (5),
mesquite-\allowbreak{}logmesquite (8),
mesquite-\allowbreak{}logmesquite\_\allowbreak{}logva (5),
mesquite-\allowbreak{}logmesquite\_\allowbreak{}logvas (8),
mesquite-\allowbreak{}logmesquite\_\allowbreak{}logvash (7),
mesquite-\allowbreak{}logmesquite\_\allowbreak{}logvolume (3),
mesquite-\allowbreak{}mesquite (8),
nes1972-\allowbreak{}nes (10),
nes1976-\allowbreak{}nes (10),
nes1980-\allowbreak{}nes (10),
nes1984-\allowbreak{}nes (10),
nes1988-\allowbreak{}nes (10),
nes1992-\allowbreak{}nes (10),
nes1996-\allowbreak{}nes (10),
nes2000-\allowbreak{}nes (10),
one\_\allowbreak{}comp\_\allowbreak{}mm\_\allowbreak{}elim\_\allowbreak{}abs-\allowbreak{}one\_\allowbreak{}comp\_\allowbreak{}mm\_\allowbreak{}elim\_\allowbreak{}abs (4),
sblrc-\allowbreak{}blr (6),
sblri-\allowbreak{}blr (6).
\endgroup


\end{document}